\documentclass[11pt,a4paper,dvipsnames]{article}
\usepackage[margin=2.2cm]{geometry}

\usepackage[utf8]{inputenc}

\usepackage{bbm}
\usepackage{tikz}
\usepackage[all,cmtip]{xy}
\usepackage[vcentermath]{youngtab}
\usepackage{ytableau}
\usepackage{amsmath}
\usepackage{amssymb}
\usepackage{amsfonts}
\usepackage{amsthm}
\usepackage{longtable}
\usepackage{mathtools}
\usepackage[none]{hyphenat}
\usepackage{setspace}
\usepackage{soul}

\usepackage[textsize=footnotesize,textwidth=1cm]{todonotes}

\usetikzlibrary{arrows.meta}
\usetikzlibrary{arrows}

\usepackage[hidelinks, pagebackref]{hyperref}
\usepackage{xcolor}
\hypersetup{
    colorlinks,
    linkcolor={blue!80!black},
    citecolor={green!50!black},
    urlcolor={red!40!black}
}

\renewcommand{\hat}[1]{\widehat{#1}}

\newcommand{\bbC}{\mathbb{C}}
\newcommand{\bbN}{\mathbb{N}}

\newcommand{\bbT}{\mathbb{T}}

\newcommand{\bbZ}{\mathbb{Z}}

\newcommand{\frakS}{\mathfrak{S}}

\newcommand{\IN}{\mathbb{N}}%
\newcommand{\IC}{\mathbb{C}}%
\newcommand{\eps}{\epsilon}%
\newcommand{\ELL}{\boldsymbol{\ell}}

\newcommand{\ELA}{\boldsymbol{a}}
\newcommand{\ELB}{\boldsymbol{b}}

\newcommand{\epsmod}[1]{\langle \epsilon^{#1}\rangle}

\newcommand{\wreath}{\wr}

\newcommand{\End}{\operatorname{End}}

\newcommand{\Stab}{\operatorname{Stab}}
\newcommand{\frakann}{\mathfrak{ann}}
\newcommand{\frakg}{\mathfrak{g}}
\newcommand{\frakgl}{\mathfrak{gl}}
\newcommand{\GL}{\operatorname{GL}}

\newcommand{\SL}{\operatorname{SL}}

\newcommand{\trace}{\mathrm{trace}}
\newcommand{\bKc}{\underline{\textup{\textsf{Kc}}}}
\newcommand{\rank}{\mathrm{rank}}

\newcommand{\vvirg}{,\ldots,}
\newcommand{\textfrac}[2]{{\textstyle \frac{#1}{#2}}}
\renewcommand{\textprod}{{\textstyle \prod}}
\newcommand{\WR}{\mathsf{WR}}
\newcommand{\bwr}{\underline{\mathsf{WR}}}
\newcommand{\uKc}{\underline{\textup{\textsf{Kc}}}}

\newcommand{\mult}{\textup{mult}}
\newcommand{\la}{\lambda}
\newcommand{\wt}{\mathrm{wt}}
\newcommand{\Sym}{S}

\newcommand{\y}{\boldsymbol{y}}
\newcommand{\bfell}{\boldsymbol{\ell}}

\newcommand{\bfx}{\mathbf{x}}
\newcommand{\x}{\boldsymbol{x}}

\newcommand{\degleq}{\trianglelefteq}
\newcommand{\Hom}{\textnormal{Hom}}%
\newcommand{\VBP}{\mathsf{VBP}}

\newcommand{\VWaring}{\mathsf{VWaring}}

\newcommand{\VNP}{\mathsf{VNP}}
\newcommand{\VF}{\mathsf{VF}}

\newcommand{\per}{\textup{per}}

\newcommand{\Kc}{\textup{\textsf{Kc}}}
\newcommand{\Kcl}{\underline{\Kc^-_{1}}}

\newcommand{\abpw}{\textup{\textsf{abpw}}}

\usepackage[nameinlink]{cleveref}
\crefname{definition}{Definition}{Definitions}
\crefname{theorem}{Theorem}{Theorems}
\crefname{lemma}{Lemma}{Lemmas}
\crefname{claim}{Claim}{Claims}
\crefname{equation}{Equation}{Equations}
\crefname{corollary}{Corollary}{Corollaries}
\crefname{proposition}{Proposition}{Propositions}
\crefname{pro}{Proposition}{Propositions}
\crefname{meta}{Meta-Question}{Meta-Questions}
\crefname{section}{Section}{Sections}
\crefname{conjecture}{Conjecture}{Conjectures}

\swapnumbers
\numberwithin{equation}{section}
\newtheorem{theorem}[equation]{Theorem}
\newtheorem{lemma}[equation]{Lemma}
\newtheorem{proposition}[equation]{Proposition}
\newtheorem{definition}[equation]{Definition}
\newtheorem{corollary}[equation]{Corollary}

\theoremstyle{definition}
\newtheorem{remark}[equation]{Remark}

\renewcommand{\bar}[1]{\overline{#1}}
\AtBeginDocument{%
\let\oldref\ref
\let\ref\cref
}
\renewcommand{\eqref}[1]{(\oldref{#1})}

\title{Geometric complexity theory for product-plus-power}

\author{Pranjal Dutta, Fulvio Gesmundo, Christian Ikenmeyer, \\ Gorav Jindal, Vladimir Lysikov}

\date{}

\newcommand{\affiliations}{
{\renewcommand{\thefootnote}{}
\footnote{\mbox{PD:} School of Computing, National University of Singapore (NUS); \texttt{pranjal@nus.edu.sg}}
\footnote{FG: Institut de Math\'ematiques de Toulouse; UMR5219 -- Universit\'e de Toulouse; CNRS -- UPS, F-31062 Toulouse Cedex 9, France; \texttt{fgesmund@math.univ-toulouse.fr}}
\footnote{CI: University of Warwick, UK; \texttt{christian.ikenmeyer@warwick.ac.uk}}
\footnote{GJ: Max Planck Institute for Software Systems, Saarbr{\"u}cken, Germany; \texttt{gjindal@mpi-sws.org}}
\footnote{VL: Ruhr-Universit\"at Bochum, Bochum, Germany; \texttt{Vladimir.Lysikov@ruhr-uni-bochum.de}}
}
}

\begin{document}
\raggedbottom

\maketitle

\thispagestyle{empty}

\begin{abstract}
According to Kumar's recent surprising result (ToCT’20), a small border Waring rank implies that the polynomial can be approximated as a sum of a constant and a small product of linear polynomials. We prove the converse of Kumar's result and establish a tight connection between border Waring rank and the model of computation in Kumar's result.
In this way, we obtain a new formulation of border Waring rank, up to a factor of the degree.

We connect this new formulation to the orbit closure problem of the product-plus-power polynomial.
We study this orbit closure from two directions:

1. We deborder this orbit closure and some related orbit closures, i.e., prove all points in the orbit closure have small non-border algebraic branching programs.

2. We fully implement the geometric complexity theory approach against the power sum by generalizing the ideas of Ikenmeyer-Kandasamy (STOC'20) to this new orbit closure. In this way, we obtain new multiplicity obstructions that are constructed from just the symmetries of the polynomials.\end{abstract}


\noindent{\footnotesize {\em 2020 Math. Subj. Class.}: 68Q17, 05E05, 05E10, 22E60}

\setcounter{page}{1}
\pagestyle{plain}

\affiliations

\setcounter{footnote}{0}

\section{Introduction}
\label{sec:intro}

Waring rank is an important complexity measure for homogeneous polynomials, classically studied in geometry and invariant theory. In complexity theory, it defines a model of computation also known as the homogeneous diagonal depth-3 circuits, see e.g.~\cite{sax08}. In addition, the \emph{matrix multiplication exponent}, a fundamental constant in computational complexity, can be described in terms of the Waring rank of a particular family of polynomials, which are a symmetrized version of the matrix multiplication tensors \cite{CHIL18}. In this work, we closely relate Waring rank and its \emph{border} version to a new measure that we call \emph{border Kumar complexity}, inspired by the recent work of Kumar \cite{kum20}.

The \emph{Waring rank} of a homogeneous degree $d$ polynomial $f$, denoted $\WR(f)$, is the {\em smallest} $r$ such that there exist homogeneous linear polynomials (also called `linear forms') $\ell_1 \vvirg \ell_r$, with $f = \sum_{i \in [r]} \ell_i^d$. This is a natural generalization of the rank of a symmetric matrix, but when $d \geq 3$, a sequence of homogeneous polynomials of Waring rank $r$ may converge to a polynomial of Waring rank strictly larger than $r$. The \emph{border Waring rank} of $f$, denoted $\bwr(f)$, is the smallest $r$ such that $f$ can be written as the limit of a sequence of polynomials $f_\eps$ with $\WR(f_\eps) =r$.

Following \cite{kum20}, we introduce another measure of complexity for homogeneous polynomials: given a homogeneous $f$ of degree $d$, the \emph{Kumar complexity} of $f$, denoted $\Kc(f)$, is the smallest
$m$ such that there exists a constant $\alpha \in \IC$ and homogeneous linear polynomials $\ell_1 \vvirg \ell_m$ with the property that
\begin{equation}\label{eq:Kc}
\tag{1.\raisebox{0.3mm}{\textreferencemark}}
\textstyle
f = \alpha\big(\prod_{i=1}^m (1+\ell_i)-1\big).
\end{equation}
For instance, given a linear form $\ell$, we see that $\Kc(\ell^d)=d$, because $\ell^d= \prod_{j=1}^d (1+\zeta^j \ell)-1$, where $\zeta$ is a primitive $d$-th root of unity. However, not all polynomials have finite Kumar complexity: for example, it is easy to see that $x_1 \cdots x_n$ cannot be expressed as in \eqref{eq:Kc}. The \emph{border Kumar complexity} of $f$, denoted $\underline{\Kc}(f)$, is the smallest $m$ such that  
\[
f = \lim_{\eps \to 0} \textstyle \alpha(\eps)\big(\prod_{i=1}^m (1+\ell_i(\eps))-1\big),
\]
for $\alpha(\eps) \in \bbC[\eps^{\pm 1}]$, and linear forms $\ell_i \in \bbC[\eps^{\pm 1}][\bfx]_1$. Alternatively, one can define $\alpha \in \IC(\eps)$ and $\ell_i \in \IC(\eps)[\x]_1$, which is equivalent, as can be seen analogously to \cite[Lem.~15.22]{burgisser2013algebraic}.
We will only be working over $\bbC[\eps^{\pm 1}]$.

Kumar \cite{kum20} proved that $\uKc(f)$ is finite for every homogeneous polynomial $f$. More precisely, he proved $\underline{\Kc}(f) \leq \deg(f) \cdot \WR(f)$, and it is easy to see that in fact $\underline{\Kc}(f) \leq \deg(f) \cdot \underline{\WR}(f)$. This inequality is not tight: for instance the square-free monomial of degree $n$ can be written as $x_1\cdots x_n = \lim_{\epsilon \to 0} \epsilon^n \left(\prod_i (1 + \frac{1}{\epsilon} x_i)-1\right)$ which shows $\uKc(x_1 \cdots x_n) = n$, whereas it is a classical fact that $\bwr(x_1 \cdots x_n) \geq \binom{n}{\lfloor n/2 \rfloor}$, which is exponentially large as a function of $n$ \footnote{We remark that determining the exact value of $\bwr(x_1 \cdots x_n)$ is an open problem. It is known that $\WR(x_1 \cdots x_n)=2^{n-1}$ \cite{CARLINI20125}. The same result on cactus rank (a scheme-theoretic version of Waring rank) is proved in \cite{ranestad2011rank}. There are (at least) two incorrect/incomplete proofs available online of the same result for border rank: the
early versions of \cite{Oed16}, and the first version of \cite{CGO:19}. A discussion on the gaps in the proofs is
available in the first version of \cite[Sec~6.1]{BB-apolarityarxiv}}. We develop a border version of the Newton identities and use it to prove a converse Kumar's result, showing that completely reducible forms are essentially the \emph{only} example where Kumar's inequality is far from being tight:
\begin{theorem}[Converse of Kumar's theorem]
\label{thm:intro:conversekumar}
Let $f$ be a homogeneous polynomial. Then either $\bwr(f) \leq \underline{\Kc}(f)$ or $f$ is a product of linear forms.
\end{theorem}

One important consequence of the close relationship between border Waring rank and Kumar complexity is a new characterization of the matrix multiplication exponent. The matrix multiplication exponent is defined as 
\[
\omega = \inf\{ \tau : \text{ two $n \times n$ matrices can be multiplied using $O(n^\tau)$ scalar multiplications}\}.
\]
This fundamental constant can be defined in terms of the tensor rank and the tensor border rank of the matrix multiplication tensor \cite{Strassen:Gauss_elimination_not_optimal,BCRL79}.
The results of \cite{CHIL18} show
that $\omega = \lim_{n \to \infty} \log_n  \bwr( \trace(X_n^3))$, where $X_n = (x_{ij})_{i,j = 1 \vvirg n}$ is a matrix of variables, so $\trace(X_n^3)$ is a homogeneous degree 3 polynomial in $n^2$ variables.
This result, together with \Cref{thm:intro:conversekumar}, implies the following 
\begin{corollary}\label{corol:kumar mamu}
    The exponent of matrix multiplication is characterized as 
    \[
    \omega = \lim_{n \to \infty} \log_n \uKc( \trace (X_n^3)).
    \]
\end{corollary}
\Cref{thm:intro:conversekumar}, together with \Cref{corol:kumar mamu}, serves as motivation to study upper and lower bounds methods for~$\uKc$. In order to study Kumar's model from a geometric point of view, we introduce a homogenized version of it, that we call \emph{product-plus-power model}. We provide several results on this model: 
\begin{itemize}
    \item We prove \emph{debordering results}, providing strong upper bounds to the border Waring rank of homogeneous polynomials in terms of their complexity in the product-plus-power model, see \Cref{subsec:prodpluspower}.
    \item We implement the geometric complexity theory program (GCT), providing an infinite family of \emph{obstructions}, see \Cref{sec:GCTintro}.
\end{itemize}

\subsection{The product-plus-power model}

Homogeneous models of algebraic computation and their border versions are conveniently described using the notions of restrictions and degenerations of homogeneous polynomials. Let $f(\bfx),g(\bfx)$ be homogeneous polynomials of degree $d$ in $n$ variables $x_1 \vvirg x_n$. We say that $g$ is a \emph{restriction} of $f$ if there exist linear forms $\ell_1(\bfx) \vvirg \ell_n(\bfx)$ such that $g(\bfx) = f(\ell_1(\bfx) \vvirg \ell_n(\bfx))$; in this case, write $g \leq f$. Equivalently, regarding $f,g$ as polynomial functions on $\bbC^n$, we have that $g \leq f$ is there exist a linear map $A : \bbC^n \to \bbC^n$ such that $g ( \bfx) = f(A^t \bfx)$. Here $A^t$ is the transpose of the matrix $A$. 

We say that $g$ is a \emph{degeneration} of $f$ is there exist linear forms $\ell_1(\bfx) \vvirg \ell_n(\bfx)$, depending rationally on a parameter $\eps$, such that $g(\bfx) = \lim_{\eps \to 0} f( \ell_1(\bfx) \vvirg \ell_n(\bfx))$; in this case, write $g \degleq f$. Equivalently, $g(\bfx) = \lim_{\eps \to 0} f(A_\eps \bfx)$ where $A_\eps$ is a linear map $A_\eps : \bbC^n \to \bbC^n$ depending on $\eps$. It is a classical fact that the entries of $A_\eps$, or equivalently the coefficients of $\ell_i(\bfx)$, can be taken in the ring of Laurent polynomials $\IC[\eps^{\pm 1 }]$, see e.g. \cite[Section 20.6]{burgisser2013algebraic}; in particular, the polynomial $f(A_\eps \bfx)$ is an element of $\bbC[\eps^{\pm 1 }, x_1 \vvirg x_n]$.

The number of variables does not play a role, and we will occasionally use the slightly larger ring $\IC[x_0,x_1,\ldots,x_n]$ for notational convenience.

Define the product-plus-power polynomial 
$
P^{[d]}_{1,1} = x_1 \cdots x_d + x_0^d \in \bbC[x_0, x_1 \vvirg x_n]_d.
$
Given a homogeneous degree~$d$ polynomial $f$ in the variables $x_1,\ldots,x_n$ (without $x_0$), it is clear that if $\uKc(f) \leq m$ then $x_0^{m-d} f$ is a degeneration of $P^{[m]}_{1,1}$. Therefore, lower bounds on the smallest possible $m$ such that $x_0^{m-d} f \degleq P^{[m]}_{1,1}$ provide lower bounds on $\uKc(f)$ and hence on~$\bwr$.

It is unclear to what extent the product-plus-power model is stronger than Kumar's model, so we provide \emph{debordering results} for the product-plus-power model, see \Cref{subsec:prodpluspower}.
We establish that this model is not much stronger, hence lower bounds are expected to exist.
To find such lower bounds on Kumar's complexity, we implement the geometric complexity theory approach for the product-plus-power model, see \ref{sec:GCTintro}.

\subsection{Debordering the product-plus-power model}
\label{subsec:prodpluspower}

Border Waring rank, and border Kumar complexity are \emph{border} measures of complexity, that is they are defined in terms of degeneration.
It is often unclear what the gap between a non-border complexity measure and its border complexity measure can be.
Lower bounds of a border complexity measure in terms of a non-border measure are commonly called \emph{debordering results} and they guarantee that introducing the degeneration procedure does not make a model of computation much stronger. We provide debordering results for the product-plus-power model and some generalizations.

We record one important debordering result for border Waring rank, in terms of the algebraic branching program width. Given a homogeneous polynomial $f$ of degree $d$, the \emph{algebraic branching program width} of $f$, denoted $\abpw(f)$, is the smallest $w$ such that $f$ can be expressed as a product $f(\bfx) = A_1(\bfx) \cdots A_d(\bfx)$ of matrices whose entries are linear forms, with $A_1$ of size $1 \times w$, $A_k$ of size $w \times w$ for $k = 2 \vvirg d-1$ and $A_d$ of size $w \times 1$. 
It is known that $\abpw(f) \leq \bwr(f)$ \cite{Forbes16,blaser2020complexity}, and further, these measures can be {\em exponentially} far (for e.g., see~\cite{ckw11}). In particular, \Cref{thm:intro:conversekumar} can be interpreted as a debordering result for Kumar's complexity, showing $\abpw(f)\leq\uKc(f)$. 

Similarly, other debordering results for border Waring rank, such as \cite[Theorem 1]{DGIJL:FixedPar}, and more generally classification results for secant varieties as in \cite{BucLan:ThirdSecantVar,ballico2017,BalBer:StratificationFourthSecantVeronese}, can be reinterpreted in terms of Kumar's complexity via \Cref{thm:intro:conversekumar}.

We obtain the following result for degenerations of the product-plus-power polynomial~$P^{[d]}_{1,1}$:
\begin{theorem}[Debordering product-plus-power]\label{thm:intro:productpluspower}
Let $f \in \IC[x_1,\ldots,x_n]_d$. 
If $f\trianglelefteq P_{1,1}^{[d]}$, then 
\begin{enumerate}
    \item[(i)] either $f\leq P_{1,1}^{[d]}$ 
    \item[(ii)] or $\bwr(f) \leq O(d^5)$.
\end{enumerate}
\end{theorem}
Also, we will obtain debordering results for a variant of the product-plus-power polynomial, that is the product-plus-two-powers polynomial $P_{1,2}^{[d]} = x_1 \cdots x_d + x_0^{d} + x_{d+1}^d$, as well. These debordering results are more fine-grained because the previous techniques would show that $\abpw(f)$ is polynomially-bounded~\cite{duttademystifying}. In \ref{subsec:lb}, we show exponential lower bounds for the same (see~\cref{thm:lowerboundprodpluspower}-\ref{thm:lowerboundproductplustwo}).

\subsection{The GCT program}\label{sec:GCTintro}
Mulmuley and Sohoni introduced the geometric complexity theory (GCT) program \cite{MS01,MS08} as an approach to Valiant's determinant vs permanent problem.
The approach aims to determine ``representation theoretic obstructions'' in order to prove lower bounds of certain complexity measures, typically by proving the impossibility of a degeneration.
Basic background knowledge on the representation theory of the general linear group and the symmetric group can for example be found in \cite{FH91} and \cite{Ful97}, and is especially valuable in \ref{sec:orbitclosure}.

The group $\GL_{n}$ has a natural action on $\IC[ x_1,\ldots,x_n]_d$ given by
$g \cdot f = f \circ g^t$ or equivalently $g \cdot f(\bfx) = f ( g^t \bfx)$ for any
$f \in  \IC[ x_1,\ldots,x_n]_d$, $g \in \GL_n$, $\bfx\in \bbC^n$.
Now, given two elements $f_1,f_2 \in \IC[ x_1,\ldots,x_n]_d$, the condition that a polynomial $f_1$ degenerates to a polynomial $f_2$ is equivalent to the fact that $f_2 \in \bar{ \GL_n \cdot f_1}$; here $\GL_n \cdot f_1 = \{g \cdot f_1 \mid g \in \GL_n\}$ denotes the orbit of $f_1$ under the group action and the overline indicates the closure, equivalently in the Zariski or the Euclidean topology of $\IC[ x_1,\ldots,x_n]_d$~\cite[AI.7.2 Folgerung]{Kra85}. Further, the membership statement $f_2 \in \bar{ \GL_n \cdot f_1}$ is equivalent to the inclusion of orbit-closures $\bar{\GL_n \cdot f_2} \subseteq \bar{ \GL_n \cdot f_1}$.

Let $\bbC\big[\IC[x_1,\ldots,x_n]_d\big]$ be the ring of polynomials over $\IC[x_1,\ldots,x_n]_d$: its elements are polynomials in the coefficients of the elements of $\IC[x_1,\ldots,x_n]_d$. Given $f \in \IC[x_1,\ldots,x_n]_d$, let $I( \bar{\GL_n \cdot f}) \subseteq \bbC\big[\IC[x_1,\ldots,x_n]_d\big]$ denote the \emph{vanishing ideal} of the algebraic variety $\bar{\GL_n \cdot f}$ and let $\bbC[\bar{\GL_n \cdot f}] = \bbC\big[\IC[x_1,\ldots,x_n]_d\big] / I( \bar{\GL_n \cdot f})$ be the quotient ring, called the \emph{coordinate ring} of $\bar{\GL_n \cdot f}$. It turns out that $\bbC[\bar{\GL_n \cdot f}]$ is a graded ring; denote by $\bbC[\bar{\GL_n \cdot f}] _D$ its degree $D$ component.
The condition $\bar{\GL_n \cdot f_2} \subseteq \bar{ \GL_n \cdot f_1}$ is equivalent to the inclusion $I( \bar{\GL_n \cdot f_1}) \subseteq I( \bar{\GL_n \cdot f_2})$, yielding a surjection of graded rings $\bbC[ \bar{ \GL_n \cdot f_1}]  \twoheadrightarrow \bbC[\bar{ \GL_n \cdot f_2}]$. In summary, if $f_1$ degenerates to $f_2$, then for every degree $D$ there is a surjective linear map 
\[
\bbC[ \bar{ \GL_n \cdot f_1}]_D  \twoheadrightarrow \bbC[\bar{ \GL_n \cdot f_2}]_D.
\]
The group action of $\GL_n$ on $\IC[x_1,\ldots,x_n]_d$ induces an action of $\GL_n$ on the coordinate rings via canonical pullback. This makes both $\IC[\overline{\GL_n f_1}]_D$ and $\IC[\overline{\GL_n f_2}]_D$ into finite dimensional representations of $\GL_n$.
The surjection is $\GL_n$-equivariant, namely it commutes with the group action.
The group $\GL_n$ is \emph{reductive}, a condition that guarantees that both representations decompose into a direct sum of \emph{irreducible representations}.
For each partition $\la$,
the \emph{representation theoretic multiplicity} $\mult_\la$ counts the number of irreducible representations of type $\la$ in such a decomposition; this number is independent of the decomposition. In the setting described here,
$\la$ is always a partition of $Dd$ into at most $n$ parts, and
Schur's Lemma (see, e.g., \cite[Lemma 1.17]{FH91}) implies an inequality of representation theoretic multiplicities: $\mult_\la(\IC[\overline{\GL_n f_2}]_D) \geq \mult_\la(\IC[\overline{\GL_n f_1}]_D)$.

The GCT program aims to prove obstructions to the existence of the surjection by determining violations of these inequalities. 
More precisely, the existence of any $\la$ with $\mult_\la(\IC[\overline{\GL_n f_1}]_D) < \mult_\la(\IC[\overline{\GL_n f_2}]_D)$ would guarantee $\overline{\GL_n f_2} \not\subseteq \overline{\GL_n f_1}$ and therefore show that $f_2$ is \emph{not} a degeneration of $f_1$. This is called a \emph{representation theoretic multiplicity obstruction}. If additionally $\mult_\la(\IC[\overline{\GL_n f_1}]_D)=0$, then we call this an \emph{occurrence obstruction}.
It is a wide open question in geometric complexity theory in which situations orbit closure containment can be disproved by representation theoretic multiplicity obstructions.
Very few examples are known \cite{BI11,BI13,IK20}.
We determine multiplicity obstructions that show that the polynomial $x_0^d + \cdots + x_d^d$ is not a degeneration of $P_{1,1}^d$ (which for example could also be seen directly by comparing the dimensions of their orbit closures):
\begin{theorem}[New obstructions]
\label{thm:intro:obstructions}
Let $d \geq 3$, and let $\la = (5d-1,1)+((d+1)\times (10d))$.
These $\la$ are representation theoretic multiplicity obstructions that show $x_0^d + \cdots + x_d^d \not\trianglelefteq P^{[d]}_{1,1}$:
\[
\mult_\la(\IC\Bigl[\overline{\GL_{d+1} P^{[d]}_{1,1}} \Bigr]) \leq 4 < 5 = \mult_\la(\IC\Bigl[\overline{\GL_{d+1} ( x_0^d + \cdots + x_d^d) }\Bigr]).
\]
\end{theorem}
This result extends the result of \cite{IK20} from the product polynomial to product-plus-power by exhibiting multiplicity obstructions that are based entirely on the symmetries of the two polynomials; see \Cref{thm:obstruction} for details. 

Indeed, from a representation-theoretic and a combinatorial point of view, the polynomial $P_{1,1}^{[d]}= \prod_{i \in [d]} x_i + x_0^d$ looks very similar to the well-studied product polynomial $P_{1,0}^{[d]} = \prod_{i \in [d]} x_i$, which was the object of several GCT papers \cite{Kum15,BI17,DIP20, IK20}. A system of set-theoretic equations for its orbit-closure was known for over a century, due to Brill and Gordon \cite{Gor94}, and their representation theoretic structure has been recently described by Guan \cite{Gua18}. We transfer as much as possible of the known theory for $P_{1,0}^{[d]}$ to the setting of $P^{[d]}_{1,1}$, in order to mimic the proof technique of \cite{IK20}. 

In order to prove \Cref{thm:obstruction}, we obtain several results about the more general polynomial $P^{[d]}_{r,s} := 
\sum_{i=1}^r \textprod_{j=1}^{d} x_{ji} + \sum_{i=1}^s y_{i}^d
$.
\begin{enumerate}
\item We determine the stabilizer of $P^{[d]}_{r,s}$ under the action of the group $\GL_{rd+s}$, see \Cref{thm:stab}.
\item We use the stabilizer to determine the representation theoretic structure of the coordinate ring of the orbit of $P^{[d]}_{1,1}$, which is achieved in \Cref{pro:branchingrules}.
\item We prove that $P^{[d]}_{r,s}$ is {\em polystable}, in the sense of invariant theory, see \Cref{pro:polystable}.
\item Polystability implies the existence of a fundamental invariant in the sense of \cite{BI17}. In \Cref{pro:alontarsi},
in the case $P^{[d]}_{1,1}$,
we show an interesting connection between the degree of this fundamental invariant and the famous Alon-Tarsi conjecture on Latin squares in combinatorics:
the fundamental invariant appears in degree $d+1$ if and only if the Alon-Tarsi conjecture holds for $d$.
\end{enumerate}

A \emph{Latin square} is an $n \times n$ matrix with entries $1,\ldots,n$ such that each row and each column is a permutation. The \emph{column sign} of a Latin square is the product of the signs of its column permutations. If $n$ is odd, then there are exactly as many sign +1 Latin squares as sign $-1$ Latin squares, and a sign-reversing involution is obtained by switching the first two rows.
The Alon-Tarsi conjecture states that for $n$ even, the number of sign +1 and sign $-1$ Latin squares are different. The main references on the Alon-Tarsi conjecture are \cite{AT:92, Dri97, Gly10},
where it is shown that the conjecture is true for $ n = p\pm 1$ for all odd primes $p$.

\subsection{Related work and context} 
Waring rank and border Waring rank are the objects of a long history of work in classical algebraic geometry and invariant theory, beginning in the nineteenth century \cite{Cay:TheoryLinTransformations,Sylv:PrinciplesCalculusForms,Cleb:TheorieFlachen}. It is related to the classical study of secant varieties \cite{Palatini:SuperficieAlg,Terr:seganti} and today it has strong connections to the study of the Gorenstein algebras \cite{Iarrobino-Kanev,BuczBucz:SecantVarsHighDegVeroneseReembeddingsCataMatAndGorSchemes} and the geometry of the Hilbert scheme of points \cite{BB-apolarity,JelMan:LimitsSaturatedIdeals}. 

In complexity theory, one is interested in the growth of a complexity parameter in a sequence of polynomials. In this context, we say that a p-family is a sequence of polynomials $(f_n)_{n \in \bbN}$ such that the degree and the number of variables of $f_n$ are polynomially bounded as functions of $n$. The complexity classes $\VWaring$ and $\bar{\VWaring}$ consist of all p-families $(f_n)_{n \in \bbN}$ such that, respectively, $\WR(f_n)$ or $\bwr(f_n)$, are a polynomially bounded function of $n$. Important complexity classes include $\VBP$ and $\VNP$, consisting of p-families with polynomially bounded determinantal complexity or permanental complexity, respectively.
The \emph{determinantal complexity} of a homogeneous polynomial $f$ of degree $d$ in $n$ variables $x_1 \vvirg x_n$ is the smallest $N$ such that $x_0^{N-d}f$ is a restriction of the determinant polynomial $\det_N = \sum_{\sigma \in \frakS_N} (-1)^\sigma \prod_{i=1}^N x_{i,\sigma(i)}$, which is a polynomial of degree $N$ in $N^2$ variables.
The \emph{permanental complexity} is defined similarly, in terms of the permanent polynomial $\per_N =   \sum_{\sigma \in \frakS_N} \prod_{i=1}^N x_{i,\sigma(i)}$. It is known that $\VBP \subseteq \VNP$ \cite{val79, Toda92}, and Valiant's determinant vs.\ permanent conjecture states that this inclusion is strict, that is $\VNP \not\subseteq \VBP$. 

The border complexity classes are $\bar{\VBP}$ and $\bar{\VNP}$ are defined similarly, replacing the notion of restriction with the one of degeneration, that is allowing arbitrary close approximations of polynomials in terms of determinants or permanents, rather than an exact expression. These border complexity classes can be defined as topological closures as well, see \cite{IS22}. A systematic study of border complexity classes was initiated in \cite{MS01,Bur04}: the Mulmuley-Sohoni conjecture is a strengthening of Valiant's conjecture, and it predicts that $\VNP \not \subseteq \bar{\VBP}$.

It is wide open whether $\VBP = \bar{\VBP}$, hence it is unclear to what extent the Mulmuley-Sohoni conjecture is stronger than Valiant's conjecture.
As a first step towards resolving this question,
in \cite{BIZ18} it is shown that border of width-2 algebraic branching programs define the same complexity class as algebraic formulas $\overline{\VBP_2} = \overline{\VF}$, over fields of characteristic $\ne 2$, and using \cite{AW16} this implies that $\VBP_2 \subsetneqq \overline{\VBP_2}$.
Resolving similar problems of inclusion between a border complexity class and a non-border complexity class is the goal of the debordering techniques. For instance, the already mentioned inequality $\abpw(f) \leq \bwr(f)$ proves the inclusion $\overline{\VWaring} \subseteq \VBP$ \cite{Forbes16,blaser2020complexity}. 
In~\cite{duttademystifying,dutta2022separated}, it is shown that $\overline{\Sigma^{[k]}\Pi\Sigma} \subsetneqq \VBP$, where $\Sigma^{[k]}\Pi\Sigma$ is the class of p-families of polynomials which can be expressed as $\sum_{i=1}^k\prod_{j}\ell_{ij}$, for linear forms $\ell_{ij}$. Very recently, \cite{duttaabp24} showed that border width-2 ABPs over characteristic $2$ can compute any algebraic formulas.

A major difficulty to achieve debordering results is that, in general, boundaries of orbits of algebraic groups may present strong geometric pathologies. In small cases, one can achieve \emph{boundary classification} results. This was done in the case of border Waring rank at most $5$ \cite{BucLan:ThirdSecantVar,BalBer:StratificationFourthSecantVeronese,ballico2017}, for the $3 \times 3$ determinant polynomial $\det_3$ \cite{huttenhain2016boundary}, and partially for the binomial $P_{2,0}^{[d]} = x_1 \cdots x_d + x_{d+1} \cdots x_{2d}$ \cite[Ch. II.9]{Hut17}. There are however universal results \cite{Kac:RmkNilpotentOrbits,vakil2006murphy,Jeli:Pathologies} hinting towards the difficulty of such a fine classification in general.

The GCT approach discussed in \Cref{sec:GCTintro} was introduced in \cite{MS01,MS08} as a path toward the Mulmuley-Sohoni conjecture, and proposed to use occurrence obstructions to prove lower bounds on the determinantal complexity of the \emph{padded} permanent $x_0^{N-n} \per_n$. The {\em no-go theorem}
of \cite{IP17,BIP19} proved that this is impossible by making use of the fact that \cite{MS01,MS08} use the padded formulation of the Mulmuley-Sohoni conjecture. There exists {\em no such result} when the determinant is replaced, for instance, by the iterated matrix multiplication polynomial, so that the determinantal complexity is replaced by the algebraic branching program width. The potential of multiplicity obstructions is explored in \cite{DIP20, IK20}: in particular, the GCT approach is used to prove that the power sum polynomial is not a product of homogeneous linear forms, although there are easier ways to prove this.

Friedman and McGuinness~\cite{FM19} give a survey about the Alon-Tarsi conjecture.
The GCT result in \cite{Kum15} is based on the conjecture.
The conjecture has been generalized in numerous directions. \cite{SW12} prove that Drisko's proof method cannot be used without modifications to prove the Alon-Tarsi conjecture.
 The same is true for results in \cite{burgisser2013explicit, BI17}, some of which are based on generalizations or variants of the conjecture.
The Polymath Project number 12 (\url{https://polymathprojects.org}) was devoted to the study of Rota's basis conjecture, which for even $n$ is implied by the Alon-Tarsi conjecture, see \cite{HR94}.
\cite{Alp17} proves an upper bound on the difference between the even and odd Latin squares.

\section{Kumar's complexity and border Waring rank}
\label{sec:Kumarscomplexity}

In this section, we prove \Cref{thm:intro:conversekumar}, connecting Waring and border Waring rank to $\Kc$-complexity and its variants. To obtain the result, we observe that $\underline{\Kc}$ expressions fall into three different cases, depending on whether the scalar $\alpha(\eps)$ in \eqref{eq:Kc} converges to~$0$, converges to a nonzero constant, or diverges. We study these three cases independently.
For the case where $\alpha(\eps)$ converges to zero, it is easy to see that the resulting polynomial is a product of affine linear polynomials, see \Cref{lem:degfdeltafKcf}.
For the case where $\alpha(\eps)$ converges to a nonzero value, we use the Newton Identities to obtain the desired lower bound given by the Waring rank, see \Cref{pro:newtonidentities}.
The case where $\alpha(\eps)$ diverges is the most interesting one as it is the one where cancellations occur in the limit; in this case the proof is obtained via a border version of Newton relations.

 Let $e_k(x_1,\hdots,x_n)$ denotes the $k$-th {\bf elementary symmetric} polynomial, defined by
\[
e_k(x_1,\hdots,x_n)\;:=\;\sum_{1 \le j_1 < j_2 < \cdots <j_{k} \le n} x_{j_1}\cdots x_{j_k}\;;
\]
Recall that by definition $e_0 = 1$. First, we record an immediate observation that will be useful throughout:
\begin{remark}\label{rmk: elementary hom poly in Kc}
It is easy to observe that
\[
\prod_{i = 1}^m (1+x_i) = \sum_{j =0}^m e_j(\bfx)
\]
where $\bfx = (x_1 \vvirg x_m)$. In particular, given a homogeneous polynomial $f \in \bbC[\bfx]_d$ of degree $d$, if $f = \alpha (\prod_{i = 1}^m (1+\ell_i) - 1)$ for homogeneous linear forms $\ell_1 \vvirg \ell_m$, then 
\begin{align*}
e_j(\ell_1 \vvirg \ell_m) &= 0 \quad \text{for all $j \neq d$}, \\
e_d(\ell_1 \vvirg \ell_m) &= \textfrac{1}{\alpha}f.
\end{align*}
\end{remark}

{\em Newton identities} are a central tool in this section; they relate the elementary symmetric polynomials and the {\em power sum} polynomial, defined as $p_k(\x):=x_1^k+\cdots+x_n^k$. 

\begin{proposition}[Newton Identities, see e.g. \cite{Macdon:SymmetricFunctions}, Section I.2] \label{prop:Newt-Id}
Let $n,k$ be integers with $n \ge k \ge 1$. Then
\[ 
k\cdot e_k(x_1,\hdots,x_n)\;=\;\sum_{i \in [k]}\,(-1)^{i-1}e_{k-i}(x_1,\hdots,x_n) \cdot p_i(x_1,\hdots,x_n)\;.
\]
\end{proposition}

In light of the \Cref{rmk: elementary hom poly in Kc}, the $\Kc$ model of computation is a sum of elementary symmetric polynomials. Shpilka \cite{Shp02} studied a similar notion of circuit complexity called $s_{sym}$. For a polynomial $f$, $s_{sym}(f)$ is defined as the smallest $m$ such that $f=e_d(\ell_1,\ell_2,\ldots,\ell_m)$ where $d=\deg(f)$ and $\ell_i$ are affine linear forms. It was proved in \cite{Shp02} that $s_{sym}(f)$ is always finite, moreover several upper and lower bounds for $s_{sym}(f)$ were proven. The complexity $\Kc$ differs from $s_{sym}(f)$, as $\Kc$ can even be infinite.
In fact, the only homogeneous polynomials with finite $\Kc$-complexity are powers of linear forms, as the following lemma shows.

\begin{lemma}\label{lem:kc-not-universal}
Let $f \in \bbC[\bfx]_d$ be a homogeneous polynomial such that $\Kc(f) < \infty$. Then $\Kc(f) = d$ and $f$ is a power of a linear form.
\end{lemma}
\begin{proof}
If $f$ is a homogeneous polynomial of degree $d$, then it is immediate that $\Kc(f) \geq \deg(f)$. Notice that for any linear form $\ell$, we have $\ell^d = \prod_{i=1}^d ( 1 + \zeta^i \ell) -1$ where $\zeta$ is a primitive $d$-th root of $1$. This shows $\Kc(\ell^d) \leq d$, hence equality holds.

Assume $f \in \bbC[x]_d$ is a homogeneous polynomial with $\Kc(f) = m < \infty$. By definition $f = \alpha \left(\prod_{i=1}^m (1+\ell_i) -1\right)$ for some homogeneous linear forms $\ell_i \in \bbC[\bfx]$. Write $\bfell = (\ell_1 \vvirg \ell_m)$. By \Cref{rmk: elementary hom poly in Kc}, we have, $e_d(\bfell) = \frac{1}{\alpha}f$ and $e_j (\bfell) = 0$ for $j \neq d$.

First, observe $m = d$. Indeed, if $m > d$, we have $0 = e_m(\bfell) = \ell_1 \cdots \ell_m$, which implies $\ell_i=0$ for some $i$, in contradiction with the minimality of $m$. Since $\Kc(f) \geq \deg(f)$, we deduce $m=d$.

Now we show that if $\bfell = (\ell_1 \vvirg \ell_d)$ satisfies $e_1(\bfell) = \cdots = e_{d-1}(\bfell) = 0$ then $e_d(\bfell) = (-1)^{d-1}\cdot \ell_d^d$; in particular, by unique factorization, all $\ell_i$'s are equal up to scaling. Write $\hat{\bfell} = (\ell_1 \vvirg \ell_{d-1})$. We use induction on $j$ to prove that $e_j(\hat{\bfell}) = (-1)^j\cdot \ell_d^j$ for $j = 1 \vvirg d-1$. For $j = 1$, we have 
\[
0 = e_1(\bfell) = (\ell_1 + \cdots + \ell_{d-1}) + \ell_d = e_1(\hat{\bfell}) + \ell_d
\]
which proves the statement. For $j = 2 \vvirg d-1$, consider the recursive relation
\[
e_j(\bfell) = e_j(\hat{\bfell}) + \ell_d e_{j-1}(\hat{\bfell}).
\]
By assumption we have $e_j(\bfell) = 0$ and the induction hypothesis guarantees $e_{j-1}(\hat{\bfell}) = (-1)^{j-1}\cdot \ell_d^{j-1}$; we deduce  $ e_j(\hat{\bfell}) = - \ell_d \cdot (-1)^{j-1}\cdot \ell_d^{j-1} = (-1)^j \ell_d^j$ which proves the statement. Finally, notice $f = \alpha e_d(\bfell) = \alpha \ell_d \cdot (-1)^{d-1} \cdot e_{d-1}(\hat{\bfell}) = -\alpha \ell_d^d$, which concludes the proof.
\end{proof}
However, the model is complete if one allows approximations, as shown in \cite{kum20}. We introduce an equivalence relation on $\bbC[\eps^{\pm 1}][\bfx]$: given two polynomials $f_1,f_2$ whose coefficients depend rationally on $\eps$, we write $f_1 \simeq f_2$ if $\lim_\eps f_1 ,\lim_\eps f_2$ are both finite and they coincide. We often use this notation with either $f_1$ or $f_2$ not depending on $\eps$: if, for instance, $f_1$ does not depend on $\eps$, then $f_1 \simeq f_2$ means that $f_2 = f_1 
+ O(\eps)$.
\begin{proposition}[\cite{kum20}]
\label{pro:introkumar}
For all homogeneous $f$ we have
$\underline{\Kc}(f) \leq \deg(f) \cdot \WR(f)$.
\end{proposition}
\begin{proof}
The proof is based on a construction by Shpilka \cite{Shp02}. Let $\WR(f) = r$ and write $f = \sum_{i=1}^r \ell_i^d$.
Let $\zeta$ be a primitive $d$-th root of unity.
Then one verifies that
\[
f = -e_{d}(-\zeta^0\ell_1,-\zeta^1\ell_1,\ldots,-\zeta^{d-1}\ell_1,\ldots \ldots, -\zeta^0\ell_r,-\zeta^1\ell_r,\ldots,-\zeta^{d-1}\ell_r)
\]
and for all $0<i<d$ we have
\[
e_{i}(-\zeta^0\ell_1,-\zeta^1\ell_1,\ldots,-\zeta^{d-1}\ell_1,\ldots \ldots, -\zeta^0\ell_r,-\zeta^1\ell_1,\ldots,-\zeta^{d-1}\ell_r) = 0.
\]
Hence $f \simeq -\eps^{-d}\big(\big((1-\eps\zeta^0\ell_1) \cdots (1-\eps\zeta^{d-1}\ell_r)\big)-1\big)$.
Therefore
$\underline{\Kc}(f) \leq rd$.
\end{proof}
In fact, the following slightly stronger statement is true:
\begin{proposition}\label{cor:KcdegWR}
For all homogeneous $f$ we have
$\underline{\Kc}(f) \leq \deg(f) \cdot \underline{\WR}(f)$.
\end{proposition}
\begin{proof}
Analogously to the proof in \ref{pro:introkumar}, let $\bwr(f) = r$ and let $\ell_1 \vvirg \ell_r$ be linear forms depending rationally on $\eps$ such that 
$f \simeq \sum_{i=1}^r \ell_i^d = -e_{d}(-\zeta^0\ell_1,\ldots, -\zeta^{d-1}\ell_r)$.
Moreover, for all $0<i<d$, we have
$e_{i}(-\zeta^0\ell_1,\ldots, -\zeta^{d-1}\ell_r) = 0$. 

Choose $M$ large enough so that for all $d < i \leq dr$ we have that $\eps^{-Md} e_i(-\eps^M\zeta^0\ell_1,\ldots, -\eps^M\zeta^{d-1}\ell_r) \simeq 0$.
We obtain
$f \simeq -\eps^{-Md}\big(\big((1-\eps^M\zeta^0\ell_1) \cdots (1-\eps^M\zeta^{d-1}\ell_r)\big)-1\big)$. Therefore
$\underline{\Kc}(f) \leq rd$.
\end{proof}
\Cref{pro:introkumar} and \Cref{cor:KcdegWR} show that if $\bwr(f)$ is small then $\bKc(f)$ is small. However, there are polynomials with large Waring (border) rank but small Kumar complexity, such as products of linear forms. For instance $\bwr(x_1 \cdots x_n)$ is exponentially large: a lower bound of $\binom{n}{\lfloor n/2 \rfloor}$ can be easily shown by partial derivative methods, see e.g.~\cite[Sec. 11]{landsberg2010}, \cite[Thm. 10.4]{ckw11}. However every completely reducible form has small Kumar complexity:
\begin{lemma}\label{lem:Kcd}
If $f = \ell_1\cdots\ell_d$ is a product of homogeneous linear forms $\ell_i$, then $\underline{\Kc}(f) = d$.
\end{lemma}
\begin{proof}
The lower bound is immediate because $\bKc(f) \geq \deg(f)$. For the upper bound, notice
$f \simeq \eps^d\big(\big(\prod_{i=1}^d (1+\eps^{-1}\ell_i)\big)-1\big)$.
\end{proof}

The main result of this section is a converse of the above statements: informally, homogeneous polynomials with small border Waring rank and product of linear forms are the only homogeneous polynomials with small border Kumar complexity.
The following result explains the relation between border Waring rank and border Kumar's complexity and completes the proof of \Cref{thm:intro:conversekumar}.
\begin{theorem}\label{thm:border-via-kumar}
If $f$ is a product of homogeneous linear forms, then $\underline{\Kc}(f)=\deg(f)$.
For all other homogeneous $f$ we have
\[
\max\{\deg(f), \ \underline{\WR}(f)\} \ \ \leq \ \ \underline{\Kc}(f) \ \ \leq \ \ \deg(f)\cdot \underline{\WR}(f). 
\]
\end{theorem}
\begin{proof}
The first statement is \ref{lem:Kcd}.
The right inequality follows from \ref{cor:KcdegWR}.
Clearly $\deg(f)\leq \underline{\Kc}(f)$.
The inequality $\underline{\WR}(f)\leq\underline{\Kc}(f)$ is a combination of \ref{lem:degfdeltafKcf},
\ref{pro:newtonidentities},
and \ref{thm:asymptoticnewtonidentities}
below.
\end{proof}

Note that in the definition of $\bKc$, the factor $\alpha$ can be assumed to be a scalar times a power of $\epsilon$, because only the lowest power of $\eps$ in $\alpha$ would contribute to the limit. We distinguish three cases, depending on the sign of the exponent of $\eps$ in $\alpha$.
\begin{itemize}
    \item $\underline{\Kc}^+(f)$ is the smallest $m$ such that $f \simeq \gamma \eps^N \big(\prod_{i=1}^m(1+\ell_i)-1\big)$ for some $N \geq 1$, $\gamma \in \bbC$ and $\ell_i \in \bbC[\eps^{\pm 1}][\bfx]_1$; set $\underline{\Kc}^+(f) = \infty$ if such an $m$ does not exist;
    \item $\underline{\Kc}^-(f)$ is the smallest $m$ such that $f \simeq \gamma \eps^{-M} \big(\prod_{i=1}^m(1+\ell_i)-1\big)$ for some $M \geq 1$, $\gamma \in \bbC$ and $\ell_i \in \bbC[\eps^{\pm 1}][\bfx]_1$; set $\underline{\Kc}^-(f) = \infty$ if such an $m$ does not exist;
    \item $\underline{\Kc}^=(f)$ is the smallest $m$ such that $f \simeq \gamma \big(\prod_{i=1}^m(1+\ell_i)-1\big)$ for some $\gamma \in \bbC$ and $\ell_i \in \bbC[\eps^{\pm 1}][\bfx]_1$; set $\underline{\Kc}^=(f) = \infty$ if such an $m$ does not exist.
\end{itemize}
We observe that
$\underline{\Kc}(f)=\min\big\{\underline{\Kc}^+(f),\ \underline{\Kc}^=(f),\ \underline{\Kc}^-(f)\big\}$.

\begin{lemma}\label{lem:degfdeltafKcf}
For all homogeneous $f$, if $\underline{\Kc}^+(f)$ is finite, then $f$ is a product of homogeneous linear forms.
\end{lemma}
\begin{proof}
Let $f \simeq \gamma\eps^{N}\big(\prod_{i=1}^m(1+\ell_i)-1\big)$ with $N \geq 1$.
Since $\eps^N\simeq 0$, we have $f  \simeq\gamma \eps^{N}\prod_{i=1}^m(1+\ell_i)$, namely $f$ is limit of a product of affine linear polynomials. The property of being completely reducible is closed, therefore we deduce that $f$ is a product of affine linear polynomials. Since $f$ is homogeneous, its factors are homogeneous as well.
\end{proof}

\begin{proposition}[Newton Identities]\label{pro:newtonidentities}
For all homogeneous $f$ we have
$\WR(f) \leq \Kc(f) = \underline{\Kc}^=(f)$.
\end{proposition}
\begin{proof}
Let $d:=\deg(f)$. Suppose $\underline{\Kc}^=(f) = m$ and write 
$f \simeq f_\eps := \gamma \big(\prod_{i=1}^m(1+\ell_i)-1\big)$. One can verify that if even one of the $\ell_i$ diverges, then the $j$-th homogeneous part of $f_\eps$ diverges, where $j$ is the number of diverging $\ell_i$.
Hence all $\ell_i$ converge and we can set $\eps$ to zero.
Hence, $\underline{\Kc}^=(f)=\Kc(f)$.
Now, since $f$ is homogeneous, each homogeneous degree $i$ part of $f_\eps$ vanishes, $i <d$.
In other words, $e_i(\ELL)=0$ for all $1 \leq i<d$, where $\ELL=(\ell_1,\ldots,\ell_m)$. Hence $s(\ELL)=0$ for \emph{all} symmetric polynomials of degree $< d$.
Therefore the Newton identity $p_d=(-1)^{d-1}\cdot d \cdot e_d + \sum_{i=1}^{d-1}(-1)^{d+i-1} e_{d-i}\cdot p_i$ gives that
$e_d(\ELL)$ and $p_d(\ELL)$ are same up to multiplication by a scalar.
Hence
$\WR(f)\leq m$.
\end{proof}

\begin{theorem}[Border Newton Identities]\label{thm:asymptoticnewtonidentities}
For all homogeneous $f$:
$\bwr(f) \leq \underline{\Kc}^-(f)$. \end{theorem}
\begin{proof}
Let $d:=\deg(f)$.
Let $f \simeq f_\eps := \gamma\eps^{-M}\big(\prod_{i=1}^m(1+\ell'_i)-1\big)$ with $M \geq 1$.
From the convergence of $f_\eps$ we deduce that for each $i$ we have $\ell'_i=\eps \ell_i$ with $\ell_i\in\IC[\eps][\x]_1$, because otherwise the homogeneous degree $j$ part diverges, where $j$ is the number of $\ell'_i$ that do not satisfy this property.

Now, let $f_{\eps,j}$ denote the homogeneous degree $j$ part of $f_\eps$.
Since $f$ is homogeneous of degree~$d$, for $0 \leq j < d$ we have $f_{\eps,j}\simeq 0$. By expanding the product, observe that for all $0 < j < d$ we have $0 \simeq f_{\eps,j} = \gamma\eps^{-M}e_j(\eps\ell_1,\ldots,\eps\ell_m) = \gamma\eps^{-M+j}e_j(\ell_1,\ldots,\ell_m)$.
We now show by induction that for all $1 \leq j < d$ we have $ \eps^{-M+j}p_j(\ell_1,\ldots,\ell_m) \simeq 0$.
This is clear for $j=1$, because $p_1=e_1$.
For the step from $j$ to $j+1$ we use Newton's identities:
\[\textstyle
p_{j+1}\;=\;(-1)^{j} \, (j+1) \, e_{j+1} + \sum_{i=1}^j (-1)^{j+i} e_{j+1-i}\cdot p_i.
\]
Hence $\eps^{-M+(j+1)}p_{j+1}(\ELL) $
\[
\hspace{10000pt minus 1fil}
=(-1)^{j} \, (j+1) \, \underbrace{\eps^{-M+(j+1)}e_{j+1}(\ELL)}_{\simeq 0} + \sum_{i=1}^j (-1)^{j+i} \underbrace{\eps^{-M+(j+1)-i} e_{j+1-i}(\ELL)}_{\simeq 0}\cdot \underbrace{\eps^M}_{\simeq 0} \cdot
\underbrace{\eps^{-M+i} p_i(\ELL)}_{\simeq 0} \simeq 0.
\hfilneg
\]
This finishes the induction proof, now we use Newton's identities again in the same way
to see that $\eps^{-M+d}p_{d}(\ELL)\simeq (-1)^{d-1} \cdot d \cdot \eps^{-M+d}e_{d}(\ELL)$:
\[
\eps^{-M+d}p_{d}(\ELL) 
=(-1)^{d-1} \cdot d \cdot \eps^{-M+d}e_{d}(\ELL) + \sum_{i=1}^{d-1} (-1)^{d-1+i} \underbrace{\eps^{-M+d-i} e_{d-i}(\ELL)}_{\simeq 0}\cdot  \underbrace{\eps^M}_{\simeq 0} \cdot
\underbrace{\eps^{-M+i} p_i(\ELL)}_{\simeq 0}.
\]
We are done now, because $f \simeq f_{\eps,d} = \gamma\eps^{-M+d} e_d(\ell_1,\ldots,\ell_m) \simeq \gamma\eps^{-M+d} \cdot \frac 1 d \cdot (-1)^{d-1} p_d(\ell_1,\ldots,\ell_m)$ and
hence $\underline{\WR}(f)\leq m$.
\end{proof}

\subsection{Linear approximations and Waring rank}\label{subsec:linearWR}
We demonstrated the inequality $\underline{\Kc}(f)\leq\deg(f)\cdot\WR(f)$
in \ref{pro:introkumar}. In the proof of \ref{pro:introkumar}, only ``linear approximations'' have been used; we prove here a converse of \ref{pro:introkumar} in the restricted setting of linear approximation. Given a homogeneous polynomial $f \in \bbC[\bfx]_d$, let $\Kcl(f)$ be the smallest $m$ such that there exist linear forms $\ell_1 \vvirg \ell_m \in \bbC[\bfx]_1$ and $M \geq 1$ such that $f\simeq\gamma\eps^{-M}\big(\prod_{i=1}^{m}(1+\eps\ell_{i})-1\big)$.
\begin{proposition}\label{pro:linearapprox}
For any homogeneous polynomial $f$ of degree $d$, we have $\WR(f)\leq\Kcl(f)\leq d\cdot\WR(f)$.
\end{proposition}

\begin{proof}
The inequality $\Kcl(f)\leq d\cdot\WR(f)$ is clear from the proof
of \ref{pro:introkumar}, as there we obtained an expression of the
form described in the definition of $\Kcl$. Suppose $\Kcl(f)=m$ and write $f\simeq f_{\eps}:=\gamma\eps^{-M}\big(\prod_{i=1}^{m}(1+\eps\ell_{i})-1\big)$
with $M\geq1$ and $\ell_{i}\in\IC[\x]_{1}$. It is immediate that $m \geq M$, $f=\gamma e_{M}(\ELL)$ and $e_j(\bfell) = 0$ for $j < M$, where $\bfell = (\ell_1 \vvirg \ell_m)$. Via the Newton identity for the power sum polynomial, we have 
\[\textstyle
{p_{M}(\ELL)=(-1)^{M-1}Me_{M}(\ELL)+\sum_{i=1}^{M-1}(-1)^{M+i-1}e_{M-i}(\ELL)\cdot p_{i}(\ELL).}
\]
Since $e_j(\ELL)=0$ for all $1\leq j<M$, we obtain:
\[\textstyle
p_{M}(\ELL)=(-1)^{M-1}Me_{M}(\ELL)=\frac{1}{\gamma}(-1)^{M-1}Mf.
\]
We conclude $\WR(f) = \WR(p_M(\bfell)) \leq \WR(p_M) = m = \Kcl(f)$, as desired.
\end{proof}

\section{Restricted binomials: debordering and lower bounds}
\label{sec:debodering}

In this section, we study restricted binomials. A binomial $\textup{bn}_d$ is the polynomial $\textup{bn}_d(\x,\y):= P_{2,0}^{[d]} = x_1 \hdots x_d + y_1 \cdots y_d$. \Cref{thm:intro:productpluspower} is based on the presentation of $P_{1,1}^{[d]}$ and $P_{1,2}^{[d]}$ as restrictions of the binomial $P_{2,0}^{[d]}$, which follows from the fact that both $x_0^d$ and $x_0^d - x_{d+1}^d = \prod_{i=1}^d (x_0 - \zeta^i x_{d+1})$ are completely reducible; here $\zeta$ is a primitive $d$-th root of $1$. Therefore, degenerations of $P_{1,1}^{[d]}$ and $P_{1,2}^{[d]}$ arise as limits of the sum of two products
\[
\lim_{\eps \to 0} \left( \prod_{i = 1}^d \ell_i(\eps) + \prod_{i = 1}^d \ell'_i(\eps)\right)
\]
where $\ell_i(\eps), \ell'_i(\eps)$ are linear forms depending rationally on $\eps$, and the second product is {\em restricted}, in the sense that, up to change of coordinates, it has either one or two variables.

In \cref{debordering-binom}, we deborder product-plus-power ($P_{1,1}^{[d]}$) and product-plus-two-powers models ($P_{1,2}^{[d]}$). In \cref{subsec:lb}, we show exponential gaps between product-plus-power, product-plus-two-powers, and binomials (in the affine sense). Identifying explicit polynomials which are hard to
approximate, and proving it remains a major template in algebraic and geometric complexity theory.
Often, proving lower bounds on the homogeneous model turns out to be easier than in its affine
model, because of the non-trivial cancellations in the latter model. However, in the restricted setting, we are able to show {\em optimal} lower bounds, see~\cref{thm:lowerboundprodpluspower} \& \ref{thm:lowerboundproductplustwo}. 

\subsection{Debordering: Characterizing special binomials} \label{debordering-binom}

In this section we prove debordering results for product-plus-power and product-plus-two-powers models.
Our method applies also for more general computational model based on restricted binomials.
More specifically, we prove that polynomials obtained in the limit in our model have low border Waring rank.
One can then apply a debordering result for $\bwr$ such as $\abpw(f) \leq \bwr(f)$~\cite{blaser2020complexity,For14} or the results from \cite{DGIJL:FixedPar} to get a complete debordering.

\begin{definition}[Restricted binomial model]
    We say that a homogeneous degree $d$ polynomial $f$ is in the \emph{class $RB_k$} if it can be presented as
    \[
    f = \prod_{i = 1}^d \ell_i + \prod_{i = 1}^d \ell'_i
    \]
    for some linear forms $\ell_i, \ell'_i$ such that $\rank(\ell'_1, \dots, \ell'_d) \leq k$.
    We also define the corresponding approximate class $\overline{RB}_k$ in the standard way: a homogeneous degree $d$ polynomial $f$ is in $\overline{RB}_k$ if
    \begin{equation}
    \label{eq:border-rb}
    f = \lim_{\varepsilon \to 0} \left(\prod_{i = 1}^d \ell_i(\eps) + \prod_{i = 1}^d \ell'_i(\eps)\right)
    \end{equation}
    for some $\ell_i(\eps), \ell'_i(\eps) \in \bbC[\epsilon^{\pm 1}][\mathbf{x}]_1$ such that $\rank(\ell'_1(\epsilon), \dots, \ell'_d(\epsilon)) \leq k$ for every $\epsilon \neq 0$.
\end{definition}
The main theorem of this section is a debordering result of $\bar{RB}_k$ in terms of border Waring rank. \Cref{thm:intro:productpluspower} is a consequence of this result.
\begin{theorem}[Debordering $\overline{RB}_k$]
\label{thm:rb-deborder}
    Let $f$ be a homogeneous polynomial of degree $d$ in $\overline{RB}_k$. Then either $f \in RB_k$, or $\bwr(f) \leq O(d^{3k + 2})$.
\end{theorem}

To prove this theorem, we first need some basic lemmas which will be used in the proof. We will use non-homogeneous polynomials, so instead of Waring rank we will be working with the complexity of $\Sigma \Lambda \Sigma$-circuits.
Denote by $\Sigma^{[s]} \Lambda^{[e]} \Sigma$ the class of (non-homogeneous) polynomials representable as a sum of $s$ powers of affine linear forms with exponents not exceeding $e$, and by $\overline{\Sigma^{[s]} \Lambda^{[e]} \Sigma}$ the corresponding class closed under approximation.
As the following lemma shows, for homogeneous polynomials this model is equal in power to border Waring rank.
\begin{lemma}\label{lem:nonhom-bwr-to-hom}
Let $f$ be a homogeneous polynomial of degree $d$. Then $f \in \overline{\Sigma^{[s]}\Lambda\Sigma}$ if and only if $\bwr(f) \le s$
\end{lemma}
\begin{proof}
Clearly, if $\bwr(f) \le s$ then $f \in \overline{\Sigma^{[s]}\Lambda\Sigma}$. For the converse, suppose $f \simeq \sum_{i \in [s]} (\alpha_i + \ell_i)^{e_i}$, where $\alpha_i \in \bbC[\eps^{\pm 1}]$, and $\ell_i \in \bbC[\eps^{\pm 1}][\x]_1$. Taking the degree $d$ part of each side, we obtain a border Waring rank decomposition $f \simeq \sum_{i \colon e_i \ge d} \binom{e_i}{d} \ell_i^d \alpha_i^{e_i-d}$ with at most $s$ summands.
\end{proof}

We recall a classical result on the border Waring rank of a binary monomial.
\begin{proposition}[see, e.g., {\cite[Cor.\,4.5]{landsberg2010}}]
\label{prop:monomial-bwr}
If $a \leq b$, then $\bwr(x^a y^b) = a + 1$.
\end{proposition}

The next lemma bounds the $\Sigma\Lambda\Sigma$ complexity of a polynomial in terms of the complexity of polynomials obtained from it by substitution of variables. 

\begin{lemma}[Interpolation] \label{lem:interpolation-border-waring}
Let $f$ be a polynomial of degree $d$ such that $f(\gamma_i, x_2,\hdots,x_n) \in \overline{\Sigma^{[s]}\wedge^{[e]}\Sigma}$ for some distinct $\gamma_0, \dots, \gamma_d \in \bbC$. Then $f \in \overline{\Sigma^{[s(d+1)^3]}\wedge^{[e + d]}\Sigma}$.
\end{lemma}
\begin{proof}
Write $f(\x)=\sum_{j=0}^d x_1^j f_j(x_2, \dots, x_n)$.
By polynomial interpolation there exist $\alpha_{ij} \in \IC$ such that $f_j = \sum_{i = 0}^d \alpha_{ij} f(\gamma_i, x_2, \hdots, x_n)$. 
By assumption, $f(\gamma_i, x_2, \hdots, x_n) \simeq \sum_{j = 1}^s \ell_{ij}^{e_j}$, where $\ell_{ij}$ are affine linear forms with coefficients in $\IC[\epsilon^{\pm 1}]$, and $e_j \le e$.
Hence 
\[
f_j(\x) \simeq \sum_{i = 0}^d \sum_{j = 1}^s \alpha_{ij} \ell_{ij}^{e_j} \implies f_j(\x) \in \overline{\Sigma^{[s(d+1)]}\wedge^{[e]}\Sigma}\;.\]
Note that for any affine linear polynomial $\ell$ the polynomial $x_1^j \ell^e$ can be approximated by a $\Sigma^{[d + 1]} \wedge^{[e + j]} \Sigma$-circuit using the decomposition of the monomial $x^j y^e$ with border Waring rank equal to $\min \{j + 1, e + 1\} \leq j + 1 \leq d + 1$; this follows from \cref{prop:monomial-bwr}.
Therefore
$x_1^j f_j \in \overline{\Sigma^{[s(d+1)^2]}\wedge^{[e + j]}\Sigma}$,
and
$f(\x) = \sum_{i = 0}^d x_1^j f_j \in \overline{\Sigma^{[s(d+1)^3]}\wedge^{[e + d]}\Sigma}$.
\end{proof}

Applying \Cref{lem:interpolation-border-waring} several times we obtain the following result.
\begin{corollary}\label{cor:interpolation-2-bwr}
Let $f(\x) \in \IC[\x]$ be a polynomial of degree $d$ such that 
\[ f(\gamma_{1i_1}, \gamma_{2i_2}, \dots, \gamma_{ki_k}, x_{k + 1},\hdots,x_n) \in \overline{\Sigma^{[s]}\wedge^{[e]}\Sigma}
\] for some $\gamma_{ij} \in \IC$, $1 \leq i \leq k$, $0 \leq j \leq d$, with $\gamma_{i0}, \dots, \gamma_{id}$ distinct for each $i$.
Then 
$ f \in \overline{\Sigma^{[s(d+1)^{3k}]}\wedge^{[e + kd]}\Sigma}.$
\end{corollary}

Additionally, we need the following statement similar to \ref{thm:asymptoticnewtonidentities}, which considers an auxiliary Kumar-like model.

\begin{theorem}
\label{thm:kumgen-twoproducts} For any degree $d$ polynomial $f(\x)\in\IC[\x]$,
not necessarily homogeneous, suppose we have $f\simeq\eps^{-M}\big(\prod_{i=1}^{m}(1+\eps a_{i})-\prod_{i=1}^{m}(1+\eps b_{i})\big)$
for some linear forms $a_{i},b_{i}\in\IC[\eps][\x]_1$ 
with $M \geq 1$.
Then $f\in\overline{\Sigma^{[2md]}\wedge^{[d]}\Sigma}$.
\end{theorem}
\begin{proof}
Let $f_{\eps}=\eps^{-M}\big(\prod_{i=1}^{m}(1+\epsilon a_{i})-\prod_{i=1}^{m}(1+ \epsilon b_{i})\big)$.
Denote by $f_{j}$ and $f_{\eps,j}$ the homogeneous degree $j$
parts of $f$ and $f_{\eps}$ respectively.
Since $f \simeq f_{\eps}$, we have
\[
f_{j} \simeq f_{\eps, j} = \eps^{-M}\left(e_{j}(\eps a_{1},\ldots,\eps a_{m})-e_{j}(\eps b_{1},\ldots,\eps b_{m})\right) = \eps^{-M+j}\left(e_{j}(\ELA)-e_{j}(\ELB)\right),
\]
where $\ELA = (a_1,\hdots,a_m)$ and similarly $\ELB =(b_1,\hdots,b_m)$.
Note that since $f_{\eps, j}$ converges, $e_{j}(\ELA)-e_{j}(\ELB)$ is divisible by $\eps^{M - j}$ for all $j \geq 1$, that is,
\[
e_{j}(\ELA)\equiv e_{j}(\ELB)\bmod\epsmod{M-j}
\]
where we consider $e_j(\ELA)$ and $e_j(\ELB)$ as elements of the ring $\bbC[\eps][\mathbf{x}]$.

We now show by induction that for all $j \geq 1$ the following additional congruences hold.
\begin{align}
p_{j}(\ELA)\equiv & p_{j}(\ELB)\bmod\epsmod{M-j}\nonumber \\
p_{j}(\ELA)-p_{j}(\ELB)\equiv & (-1)^{j-1}j\left(e_{j}(\ELA)-e_{j}(\ELB)\right)\bmod\epsmod{M-j+1}\nonumber 
\end{align}
The case $j=1$ is trivially true because $p_{1}=e_{1}$.
For the induction step from $j$ to $j+1$, we use Newton's identities
\[
{\textstyle p_{j+1}=(-1)^{j}(j+1)e_{j+1}+\sum_{i=1}^{j}(-1)^{j+i}e_{j+1-i}\cdot p_{i}.}
\]
We obtain 
\begin{align}
p_{j+1}(\ELA)-p_{j+1}(\ELB)= & (-1)^{j}(j+1)\left(e_{j+1}(\ELA)-e_{j+1}(\ELB)\right)\notag\\
\label{eq:newton-diff}
+ & \sum_{i=1}^{j}(-1)^{j+i}\left(e_{j+1-i}(\ELA)\cdot p_{i}(\ELA)-e_{j+1-i}(\ELB)\cdot p_{i}(\ELB)\right).
\end{align}
By induction hypothesis we know that for $1\leq i\leq j$
\begin{align*}
p_{i}(\ELA)\equiv & p_{i}(\ELB)\bmod\epsmod{M-i}\\
e_{j+1-i}(\ELA)\equiv & e_{j+1-i}(\ELB)\bmod\epsmod{M-(j+1)+i}.
\end{align*}
Since $M - j \leq M - i$ and $M - j \leq M - (j + 1) + i$, this can be relaxed to
\begin{align*}
p_{i}(\ELA)\equiv & p_{i}(\ELB)\bmod\epsmod{M-j}\\
e_{j+1-i}(\ELA)\equiv & e_{j+1-i}(\ELB)\bmod\epsmod{M-j}.
\end{align*}
From \eqref{eq:newton-diff} we get
\[
p_{j+1}(\ELA)-p_{j+1}(\ELB)\equiv(-1)^{j}(j+1)\left(e_{j+1}(\ELA)-e_{j+1}(\ELB)\right)\bmod\epsmod{M-j}.
\]
Weakening this to an equivalence $\bmod\epsmod{M-(j+1)}$, we obtain
\[
p_{j+1}(\ELA)-p_{j+1}(\ELB)\equiv(-1)^{j}(j+1)\left(e_{j+1}(\ELA)-e_{j+1}(\ELB)\right)\equiv 0\bmod\epsmod{M-(j+1)},
\]
or $p_{j+1}(\ELA)\equiv p_{j+1}(\ELB)\bmod\epsmod{M-(j+1)}$, finishing the induction.

Finally, we use the proved congruences to write an approximate decomposition of $f$. We have
\[
f_{j} \simeq\eps^{-M+j}\left(e_{j}(\ELA)-e_{j}(\ELB)\right)\simeq \eps^{-M+j}\cdot\frac{1}{j}\cdot(-1)^{j-1}\left(p_{j}(\ELA)-p_{j}(\ELB)\right),
\]
which shows that $\bwr(f_j) \leq 2m$.
Note that $f_0 = 0$, so $f = \sum_{j = 1}^d f_j \in \overline{\Sigma^{[2md]}\Lambda^{[d]}\Sigma}$.
\end{proof}

\begin{corollary}
\label{cor:kumgen-twoproducts} For any degree $d$ polynomial $f(\x)\in\IC[\x]$,
not necessarily homogeneous, suppose we have $f\simeq\eps^{-M}\big(\alpha\prod_{i=1}^{m}(1+\eps a_{i})-\beta\prod_{i=1}^{m}(1+\eps b_{i})\big)$
with $M \geq 1$ 
for some $a_{i},b_{i}\in\IC[\eps][\x]_1$ and $\alpha,\beta\in\IC[\eps]$ such that $\alpha \simeq \beta \not\simeq 0$.
Then $f\in\overline{\Sigma^{[2md 
+ 1]}\wedge^{[d]}\Sigma}$.
\end{corollary}
\begin{proof}
Let $f_j$ and $f_{\eps, j}$ be the homogeneous parts as in the proof of the preceding Theorem.
Additionally, Let $\alpha_0 = \lim_{\eps\to 0} \alpha$ and $\gamma' = \frac{\beta}{\alpha} \in \IC[[\eps]]$. As mentioned earlier, one can truncate and work with $\gamma \equiv \gamma'~\bmod~\langle \eps^{r}\rangle$, for some large positive integer $r$.
From assumptions of the theorem, $\alpha_0 \neq 0$ and $\gamma \simeq 1$.
We have 
\[
\frac{1}{\alpha_0} f \simeq \frac{1}{\alpha} f \simeq \eps^{-M}\big(\prod_{i=1}^{m}(1+\eps a_{i})-\gamma\prod_{i=1}^{m}(1+\eps b_{i})\big)
\]
By taking degree $0$ part we get $\frac{1}{\alpha_0} f_0 \simeq \frac{1}{\alpha_0} f_{\eps, 0} = \eps^{-M} (1 - \gamma)$, so for $j \geq 1$ we have
\[
\frac{1}{\alpha_0} f_j \simeq  \eps^{-M + j} (e_j(\ELA) - \gamma e_j(\ELB)) = \eps^{-M + j} (e_j(\ELA) - e_j(\ELB)) + \eps^j \frac{f_{\eps, 0}}{\alpha_0} e_j(\ELB) \simeq \eps^{-M + j} (e_j(\ELA) - e_j(\ELB)),
\]
hence
\[
f \simeq f_0 + \alpha_0 \eps^{-M} \big(\prod_{i=1}^{m}(1+\eps a_{i})-\prod_{i=1}^{m}(1+\eps b_{i})\big),
\]
and we reduce to the case considered in \ref{thm:kumgen-twoproducts}.
\end{proof}

We are now ready to prove \ref{thm:rb-deborder}.

\begin{proof}[Proof of~\ref{thm:rb-deborder}]
Since $f \in \overline{RB}_k$, it has an approximate decomposition~\eqref{eq:border-rb}, which we rewrite as
\[
f \simeq \eps^{p} \prod_{i = 1}^d \ell_i - \eps^{p'}\prod_{i = 1}^d \ell'_i
\]
where $\ell_i, \ell'_i \in \bbC[\epsilon][\mathbf{x}]_1$ are not divisible by $\eps$ and $\rank(\ell'_1, \dots \ell'_d) \leq k$ at any $\eps \neq 0$.
Define $\ell_{i0} \in V$ as $\ell_{i0} = \ell_{i}|_{\eps = 0}$ and similarly $\ell'_{i0} = \ell'_{i}|_{\eps = 0}$.
$\ell_{i0}$ and $\ell'_{i0}$ are nonzero and by semicontinuity of rank we have $\rank(\ell'_{10}, \dots \ell'_{d0}) \leq k$.

If $p = p' = 0$, then $f = \prod_{i = 0}^d \ell_{i0} - \prod_{i = 0}^d \ell'_{i0}$.
Similarly, if one of the exponents $p$ and $p'$ is positive, then the corresponding summand tends to $0$ as $\eps \to 0$, and $f$ is a product of linear forms, and if both $p$ and $p'$ are positive, then $f = 0$.
In all these cases we have $f \in RB_k$.

Consider now the case when there are negative exponents.
The convergence of the right hand side of the decomposition implies that $p = p'$ and the lowest degree term $\prod_{i = 0}^d \ell_{i0} - \prod_{i = 0}^d \ell'_{i0}$ is zero.
By unique factorization the sets of linear forms $\ell_{i0}$ and $\ell'_{i0}$ are the same up to scalar multiples, and we can permute and rescale the factors in one of the products so that $\ell_{i0} = \ell'_{i0}$.
Additionally we can assume that $\ell_{10}, \dots, \ell_{r0}$ are linearly independent, where $r = \rank(\ell_{10}, \dots \ell_{d0}) \leq k$.

Since $\ell_{i0}$ for $i \leq r$ are linearly independent, there exists an invertible linear map $A$ such that $\ell_{i0}(A\mathbf{x}) = x_i$ for $i \leq r$.
The linear forms $\ell_{i0}$ lie in the linear span of the first $r$ of them, which means that $\ell_{i0}(A\mathbf{x}) \in \IC[x_1, \dots, x_r]_1$ for all $i$.

Let $M = -p$, $L_i(\mathbf{x}) = \ell_i(A\mathbf{x})$ and $L'_i(\mathbf{x}) = \ell'_i(A\mathbf{x})$.
For the polynomial $g(\mathbf{x}) = f(A\mathbf{x})$ we obtain an approximate decomposition of the following form
\[
g \simeq \eps^{-M} \big(\prod_{i = 1}^d L_i - \prod_{i = 1}^d L'_i\big)
\]
where $L_i, L'_i \in \IC[\eps][\mathbf{x}]_1$ are such that $L_{i0} := L_i|_{\eps = 0} = L'_i|_{\eps = 0}$ are nonzero elements of $\IC[x_1, \dots, x_r]$.

Choose $\gamma_{ij} \in \IC$ for $1 \leq i \leq r$, $0 \leq j \leq d$ so that $\gamma_{i0}, \dots, \gamma_{id}$ are distinct for each $i$ and $L_{i0}(\gamma_{1j_1}, \dots, \gamma_{rj_r}) \neq 0$ for all $i, j_1, \dots, j_r$.
The choice is possible because $L_{k0}$ are nonzero and hence the set of tuples $\gamma$ not satisfying the required conditions is a nontrivial Zariski closed set.
Write 
\begin{align*}
L_i(\gamma_{1j_1}, \dots, \gamma_{rj_r}, x_{r + 1}, \dots, x_n) &= \alpha_i + \eps A_i(x_{r + 1}, \dots, x_n)\\
L'_i(\gamma_{1j_1}, \dots, \gamma_{rj_r}, x_{r + 1}, \dots, x_n) &= \beta_i + \eps B_i(x_{r + 1}, \dots, x_n)
\end{align*}
with $\alpha_i, \beta_i \in \IC[\eps]$ such that $\alpha_i \simeq \beta_i$ and $A_i, B_i \in \IC[\eps][x_{r + 1}, \dots, x_n]_1$.
Set $\alpha = \prod_{i = 1}^d \alpha_i$, $\beta = \prod_{i = 1}^d \beta_i$, $a_i = \frac{A_i}{\alpha_i}$, $b_i = \frac{B_i}{\beta_i}$. Because $\alpha_i|_{\eps = 0} = L_{i0}(\gamma_{1j_1}, \dots, \gamma_{rj_r}) \neq 0$, $a_i$ are well defined in the ring $\IC[[\eps]][\mathbf{x}]$; ditto for $b_i$. As argued earlier, truncating and working with finite precision of $\eps$ suffices, therefore, let $a_i' := a_i \bmod~\langle \eps^r \rangle$, and similarly for $b_i'$, for some large positive integer $r$.
We obtain
\[
g(\gamma_{1j_1}, \dots, \gamma_{rj_r}, x_{r + 1}, \dots, x_n) \simeq \eps^{-M} \big(\alpha \prod_{i = 1}^d (1 + \eps a_i') - \beta \prod_{i = 1}^d (1 + \eps b_i')\big).
\]
By~\ref{cor:kumgen-twoproducts} $g(\gamma_{1j_1}, \dots, \gamma_{rj_r}, \mathbf{x}) \in \overline{\Sigma^{[2d^2 + 1]} \Lambda^{[d]} \Sigma}$.
By~\ref{lem:interpolation-border-waring} $g \in \overline{\Sigma^{[(2d^2 + 1)(d + 1)^{3r}]} \Lambda^{[(r + 1)d]} \Sigma}$, and by~\ref{lem:nonhom-bwr-to-hom} $\bwr(g) \leq (2d^2 + 1)(d + 1)^{3r} = O(d^{3k + 2})$.
Since border Waring rank is invariant under invertible linear transformations, the same is true for $f$.
\end{proof}

As special cases we obtain the following results for product-plus-power and product-plus-two powers.
Note that $RB_1$ consists of polynomials of the form $\prod_{i = 1}^d \ell_i + {\ell'}_1^{d}$, which are exactly the restrictions of $P^{[d]}_{1,1}$.
Similarly, $f \in \overline{RB}_1$ if and only if $f \trianglelefteq P^{[d]}_{1,1}$.
Therefore, \Cref{thm:intro:productpluspower} is a consequence of \Cref{thm:rb-deborder}.

Similarly, the result for the product-plus-two-powers follows for the analysis of $\overline{RB}_2$, because the sum of two powers $x_0^d - x_{d+1}^d$ can be represented as a product of linear forms in two variables $x_0^d - x_{d+1}^d = \prod_{i = 1}^d (x_0 - \zeta^i x_{d+1})$, where $\zeta$ is a primitive $d$-th root of unity. A more careful case-by-case analysis gives the following result.

\begin{theorem}[debordering product-plus-two-powers] \label{thm:product+two-powers}
Let $f \in \IC[x_1,\cdots,x_n]_d$ such that $f \trianglelefteq P_{1,2}^{[d]}$. 
One of the three alternatives is true:
\begin{enumerate}
    \item $f \leq P_{1,2}^{[d]}$, or
    \item $f \leq \prod_{i = 1}^d \,y_i + y_0^{d-1} \cdot y_{d+1}$, or
    \item $\bwr(f) = O(d^8)$.
\end{enumerate} 
\end{theorem}
\begin{proof}
The proof mostly follows the proof of~\ref{thm:rb-deborder}. For the completeness, we give a detailed proof.

Since, $f \trianglelefteq P_{1,2}^{[d]}$, by definition,~$f = \lim_{\varepsilon \to 0} \left(A + B + C\right)$, where $A := \prod_{i = 1}^d \ell_i(\varepsilon)$, $B := \ell'(\varepsilon)^d$, and $C := \ell''(\varepsilon)^d$. There are a few cases to analyze.

\begin{enumerate}
    \item[(I)] If individually, $\lim_{\varepsilon \to 0} A$, $\lim_{\varepsilon \to 0} B$, and $\lim_{\varepsilon \to 0} C$ exist, then $f \leq P_{1,2}^{[d]}$. 
    \item[(II)] If $g:= \lim_{\varepsilon \to 0} (A+B)$ and $h:= \lim_{\varepsilon \to 0} C$ exist, then note that $g\trianglelefteq P_{1,1}^{[d]}$, then by~\cref{thm:intro:productpluspower}, we have
\begin{enumerate}
    \item[(i)] either $g\leq P_{1,1}^{[d]}$ 
    \item[(ii)] or $\bwr(g) \leq O(d^5)$.
\end{enumerate}
Since, $\bwr(h)=1$, and $f = g+h$, the theorem follows.
\item[(III)] If $g:= \lim_{\varepsilon \to 0} A$ and $h:= \lim_{\varepsilon \to 0} (B+C)$ exist, then note that $g\leq P_{1,0}^{[d]}$, and $\bwr(h)=2$. It is known that this either $\WR(h) = 2$ or $h = \hat{\ell}_1^{d-1} \hat{\ell}_2$ for two linear forms $\hat{\ell}_1,\hat{\ell}_2$, see, e.g., \cite{landsberg2010}. Therefore, $f= g+h$ corresponds to either (1) or (2).
\item[(IV)] If none of (I)--(III) is true, then
note that $f$ can be rewritten as
\[\lim_{\varepsilon \to 0} \left(\prod_{i=1}^d \ell_i(\varepsilon) + \prod_{i=1}^d (\ell'(\varepsilon) - \zeta^{2i-1} \ell''(\varepsilon))\right),\]
where $\zeta$ is the $(2d)$-th primitive root of unity. Further, it is easy to see that $\rank\left(\ell'(\varepsilon) - \zeta \ell''(\varepsilon), \cdots, \ell'(\varepsilon) - \zeta^{2d-1} \ell''(\varepsilon)\right) \le 2$, for every $\varepsilon \ne 0$. Therefore, by definition~$f \in \overline{RB}_2$. Using, \cref{thm:rb-deborder}, we get that $\bwr(f) \le O(d^8)$. 
\end{enumerate}
This finishes the proof.
\end{proof}

\subsection{Lower Bounds}\label{subsec:lb}

In this section, we prove several exponential separations between related polynomials contained in the affine closure of binomials.

\begin{lemma}\label{lem:irreducible-1}
The polynomial $P^{[d]}_{1,2} = \prod_{i \in [d]} x_i + x_{d+1}^d + x_{d+2}^d$ cannot be written as a product of linear forms.
\end{lemma}

\begin{proof}
For every homogeneous polynomial $f$ of degree $d$ which is a product of linear forms, the space of first order partial derivatives has dimension at most $d$. But $\prod_{i \in [d]} x_i + x_{d+1}^d + x_{d+2}^d$ clearly has $d+2$ linearly independent partial derivatives.
\end{proof}

\begin{lemma}
\label{lem:irreducible-2}
The polynomial $P^{[d]}_{2,0} = \prod_{i=1}^d x_i + \prod_{i=d+1}^{2d} x_i$ cannot be written as a product of linear forms.
\end{lemma}
\begin{proof}
It easily follows from a proof similar to that of \cref{lem:irreducible-1}.
\end{proof}

\begin{lemma}\label{lem:topfaninlowerboundf}
For the polynomial $P^{[d]}_{1,2}=\prod_{i\in[d]}x_{i}+x_{d+1}^{d}+x_{d+2}^{d}$, we have $\bwr(f)\geq 2^{\Omega(d)}$.
\end{lemma}
\begin{proof}
Evaluating $x_{d+1}=x_{d+2}=0$, we obtain
\[
\bwr(f) \geq \bwr(x_1 \cdots x_d) \geq \binom{d}{\lceil d/2 \rceil}
\]
where the second inequality follows computing the dimension of the space of partial derivatives of order $\lfloor d/2\rfloor$, see, e.g.,~\cite[Prop.\,11.6]{landsberg2010}.
\end{proof}

For polynomials $f$ and $g$ (not necessarily homogeneous) we write 
$f\leq_{\textup{aff}}g$
if there exists an affine linear map $A$ with $f = g\circ A$.
We write $f\trianglelefteq_{\textup{aff}}g$
if there exist affine linear maps $A_\eps$ with $f = \lim_{\eps \to 0} g \circ A_\eps$.

\begin{theorem}[First exp.~gap theorem]\label{thm:lowerboundprodpluspower}
If  $P_{1,2}^{[d]} \trianglelefteq_{\textup{aff}}  P^{[e]}_{1,1}$, then $e \ge \exp(d)$.
\end{theorem}

We remark that by Kumar's result \cite{kum20}, we know that there exists $e \le \exp(d)$, such that $P_{1,2}^{[d]} \trianglelefteq_{\textup{aff}}  P^{[e]}_{1,1}$. Therefore, \cref{thm:lowerboundprodpluspower} is {\em optimal}.

\begin{proof}[Proof of \cref{thm:lowerboundprodpluspower}]
Let~$P_{1,2}^{[d]} \trianglelefteq_{\textup{aff}}  P^{[e]}_{1,1}$. That means that there are affine linear forms $L_i \in \bbC[\eps^{\pm 1}][\x]$ such that $\prod_{i \in [d]} x_i + x_{d+1}^d + x_{d+2}^d + \epsilon \cdot S(\x,\epsilon)= \prod_{i \in [e]} L_i + L_{e+1}^e$. 
By substituting, $x_i \mapsto x_i/x_0$, and multiplying both sides by $x_0^e$, we get that $x_0^{e-d} \cdot P_{1,2}^{[d]} + \epsilon \cdot \hat{S}= \prod_{i \in [e]} \hat{L}_i + \hat{L}_{e+1}^e$, for homogeneous linear forms $\hat{L}_i$, or, equivalently, $x_0^{e-d} \cdot P_{1,2}^{[d]} \trianglelefteq P^{[e]}_{1,1}$. 

By \ref{thm:intro:productpluspower}, we know that $x_{0}^{e-d}\cdot P_{1,2}^{[d]} \trianglelefteq P^{[e]}_{1,1}$ implies either (i) $x_{0}^{e-d}\cdot P_{1,2}^{[d]} = \prod_{i \in [e]} \ell_i + \ell_0^e$, for some linear forms $\ell_i \in \IC[\x]$, or (ii) $\bwr(x_{0}^{e-d}\cdot P_{1,2}^{[d]})= O(e^5)$. We show that (i) is an impossibility while (ii) can happen only when $e \ge \exp(d)$. 

\medskip
{\bf \noindent Proof of Part (ii):}  Fix a random $x_0=\alpha \in \IC$. Note that, this implies that  $P_{1,2}^{[d]}+\epsilon g=\sum_{i\in[k]}\ell_{i}^{e}$ for some
affine forms $\hat{\ell}_{i}\in\bbC[\eps^{\pm 1}][\x]$ and $g\in\IC[\epsilon][\x]$ with $k\in O(e^5)$. Since $P_{1,2}^{[d]}$ is homogeneous, this also implies that $\bwr(P_{1,2}^{[d]})\leq k$.
But then \ref{lem:topfaninlowerboundf} implies that $k\geq 2^{\Omega(d)}$, which in turn implies that $e\geq 2^{\Omega(d)}$.

\medskip
{\bf \noindent Proof of Part (i):} Let $x_0^{e-d} \cdot P_{1,2}^{[d]} = \prod_{i \in [e]} \ell_i + \ell_0^e$. The space of first order partials of the LHS has dimension at least $d+2$, while the one of the RHS has dimension at most $e+1$; since trivially $\prod_{i \in T} \ell_i$, for $T \subset [e]$, such that $|T|=e-1$, and $\ell_0^{e-1}$ certainly span the space of single partial derivatives. Therefore, $e \ge d+1$.  This will be important since we will use the fact that $e-d \ge 1$, in the below.

Further, we can assume that $x_0 \nmid \ell_0$. Otherwise, say $\ell_0 = c \cdot x_0$, for some $c \in \IC$, which implies that $x_0^{e-d} \mid \prod_{i \in [e]} \ell_i$. Hence, without loss of generality, we can assume that $\ell_i = x_0$, for $i \in [e-d]$ (we are assuming constants to be $1$, because we can always rescale and push the constants to the other linear forms). Therefore, RHS is divisible by $x_0^{e-d}$. By dividing it out and renaming the linear forms appropriately, we get 
\[ 
P_{1,2}^{[d]}\;=\; \prod_{i \in [d]} \hat{\ell}_i + c x_0^d\;,
\]
where $\hat{\ell}_i \in \IC[\x]$. Further, we can put $x_0=0$. Note that, $x_0 \nmid \hat{\ell}_i$, for any $i$, since otherwise $x_0$ divides RHS, but it doesn't divide the LHS. After substituting $x_0=0$, we get that 
\[ 
P_{1,2}^{[d]} \;=\; \prod_{i \in [d]} \tilde{\ell}_i\;,
\]
where $\IC[x_1,\hdots,x_{d+2}] \ni \tilde{\ell}_i = \hat{\ell}_i \rvert_{x_0=0} \ne 0$.  From \ref{lem:irreducible-1}, it follows that this is not possible. A similar argument shows that $x_0 \nmid \ell_i$, for any $i \in [d]$; because otherwise that implies $x_0 \mid \ell_0$, and hence the above argument shows a contradiction.

Therefore, we assume that $x_0 \nmid \ell_i$, for $i \in [0,d]$. Now, there are two cases -- (i) $x_0$ appears in $\ell_0$, (ii) $x_0$ does not appear in $\ell_0$. 

If $x_0$ appears in $\ell_0$, then say $\ell_0 = c_0 x_0 + \hat{\ell}_0$, for some $c_0 \neq 0$. Note that $\hat{\ell}_0 \in \IC[x_1,\hdots,x_{d+2}]_1$ is non-zero, since we assume that $x_0 \nmid \ell_0$. Substitute $x_0 = - \hat{\ell}_0/c_0$ (so that $\ell_0$ vanishes). This implies:
\[ 
(- \hat{\ell}_0/c_0)^{e-d} \cdot P_{1,2}^{[d]} \;=\; \prod_{i \in [e]}\,\hat{\ell}_i\;,
\] 
where $\hat{\ell}_i = \ell_i\rvert_{x_0=- \hat{\ell}_0/c_0}$. Since LHS is non-zero, so is each $\hat{\ell}_i$. Since, everything is homogeneous, and we have unique factorization, the above implies that up to renaming, $P_{1,2}^{[d]}= c \cdot \prod_{i \in [d]} \hat{\ell}_i$, which is a contradiction by \ref{lem:irreducible-1}.

If $x_0$ does not appear in $\ell_0$, then there must exist an $i \in [e]$ such that $x_0$ appears in $\ell_i$, otherwise RHS is $x_0$-free which is trivially a contradiction. We also know that $x_0$ cannot divide $\ell_i$, by our assumption. So, say $\ell_i = c_i x_0 + \hat{\ell}_i$, where $\hat{\ell}_i$ is $x_0$-free, and $c_i \in \IC$ is a nonzero element. Substitute $x_0 = -\hat{\ell}_i/c_i$, so that $\ell_i$ vanishes. Since $\ell_0$ is $x_0$-free, we immediately get that
\[ 
(- \hat{\ell}_0/c_0)^{e-d} \cdot P_{1,2}^{[d]} \;=\; \ell_0^e\;.
\]
Again, by unique factorization, we get that $P_{1,2}^{[d]} = c \cdot \ell_0^d$, for some $c \in \IC$, which is a contradiction by \ref{lem:irreducible-1}. This finishes the proof.
\end{proof}

\begin{theorem}[Second exp.~gap theorem]\label{thm:lowerboundproductplustwo}
If $P_{2,0}^{[d]} \trianglelefteq_{\textup{aff}}  P_{1,2}^{[e]}$, then $e \ge \exp(d)$.
\end{theorem}

\begin{proof}
Let~$P_{2,0}^{[d]} \trianglelefteq_{\textup{aff}}  P_{1,2}^{[e]}$. A similar formulation as above (in the previous theorem) gives us that
$x_{0}^{e-d}\cdot P_{2,0}^{[d]}  \trianglelefteq P^{[e]}_{1,2}$.
By \ref{thm:product+two-powers}, we know that $x_{0}^{e-d}\cdot P_{2,0}^{[d]} \trianglelefteq  P_{1,2}^{[e]}$ implies -- either (i) $x_{0}^{e-d}\cdot P_{2,0}^{[d]} = g+h$, where $g=\prod_{i \in [e]} \ell_i$, for some linear forms $\ell_i \in \IC[\x]_1$, and $\bwr(h)\le 2$, or (ii) $\bwr(x_{0}^{e-d}\cdot P_{2,0}^{[d]})= O(e^8)$. Similarly, as before, we  show that (i) is an impossibility while (ii) can happen only when $e \ge \exp(d)$. Part (ii) proof is exactly to the argument in the proof of \ref{thm:lowerboundprodpluspower}.

To prove the Part (i), there are two cases -- (a) $h= \ell_0^{e} + \ell_{e+1}^{e}$, for $\ell_i \in \IC[\x]_1$, or, (b) $h= \ell_0^{e-1} \cdot \ell_{e+1}$. 

\medskip
{\bf \noindent Case (a):} Let $x_{0}^{e-d}\cdot P_{2,0}^{[d]} = \prod_{i \in [e]}\,\ell_i + \ell_0^{e} + \ell_{e+1}^{e}$. We assume that $x_0$ does not divide~$\ell_i$, for some $i \in \{0,e+1\}$, and each $\ell_i$, for $i \in [e]$, otherwise, we can divide by the maximum power of $x_0$ on both the sides.  

The space of first order partial derivatives of the LHS has dimension at least $2d$ whereas the one of the RHS has dimension at most $e+2$: since trivially $\prod_{i \in T} \ell_i$, for $T \subset [e]$, such that $|T|=e-1$, and $\ell_0^{e-1}, \ell_{e-1}^e$ certainly span the space of single partial derivatives. Therefore, $e \ge 2d-2$. 

Now, we divide this into subcases:

\medskip
(a1) $x_0$ does not appear in $\ell_i$, for any $i \in [e]$,

(a2) $x_0$ appears in $\ell_i$, for some $i \in [e]$.

\medskip
\noindent {\bf Case (a1):} $x_0$ does not appear in $\ell_i$, for $i \in [e]$. In that case, say $\ell_0 = c_0 x_0 + \hat{\ell}_0$, and $\ell_{e+1} = c_{e+1} x_0 + \hat{\ell}_{e+1}$, where $\hat{\ell}_0$ and $\hat{\ell}_{e+1}$ are $x_0$-free, and $c_0, c_{e+1}$ are constants (might be $0$ as well, but both cannot be $0$ since then RHS becomes $x_0$-free). Therefore, the coefficient of $x_0^{e-d}$ (as a polynomial) in RHS is $\gamma_0 \hat{\ell}_0^{d} + \gamma_{e+1} \hat{\ell}_{e+1}^d$, where $\gamma_0 = \binom{e}{d} c_0^{e-d}$, and similarly $\gamma_{e+1}= \binom{e}{d} c_{e+1}^{e-d}$. Comparing with LHS, we get that $P_{2,0}^{[d]}=\gamma_0 \hat{\ell}_0^{d} + \gamma_{e+1} \hat{\ell}_{e+1}^d$. Trivially, over $\IC$, $\gamma_0 \hat{\ell}_0^{d} + \gamma_{e+1} \hat{\ell}_{e+1}^d$ is a product of linear forms, which is a contradiction, using \cref{lem:irreducible-2}.

\medskip
\noindent {\bf Case (a2):} If $x_0$ appears in one of the $\ell_i$ for $i\in [e]$, it can appear in two ways, either $\ell_i$ is a constant multiple of $x_0$, or $\ell_i= c_ix_0 + \hat{\ell}_i$, where $\hat{\ell}_i$ is a nonzero linear form which is $x_0$-free. Let $S_1 \subseteq [e]$ be such that  $\ell_i = c_i \cdot x_0$, for $i \in S_1$, for some nonzero constant $c_i \in \IC$, and $S_2 \subseteq [e]$ be such that $\ell_i =c_ix_0 + \hat{\ell}_i$, where $\hat{\ell}_i$ is nonzero.

Note that if $|S_1| + |S_2| < e-d$, then $x_0^{e-d}$ cannot be contributed from the product and hence it only gets produced from $\ell_0^{e} + \ell_{e+1}^e$, and we get a contradiction in the same way as above. Hence, without loss of generality, assume that $|S_1| + |S_2| \ge e-d$. 

If $S_2$ is non-empty, say $j \in S_2$, then substitute $x_0 = - \hat{\ell}_j/c_j$, so that $\ell_j$ becomes $0$. This substitution gives us the following:
\[ 
(-\hat{\ell}_j/c_j)^{e-d} \cdot P_{2,0}^{[d]}\;=\;\tilde{\ell}_0^{e} + \tilde{\ell}_{e+1}^e\;.
\]
Since, $\tilde{\ell}_0^{e} + \tilde{\ell}_{e+1}^e$ can be written as a product of linear forms, from the unique factorization, it follows that $f$ must be a product of linear forms, which is a contradiction from \cref{lem:irreducible-2}. Hence, we are done when $|S_2|$ is non-empty.

If $S_2$ is empty, since $|S_1| + |S_2| \ge e-d$ by assumption, we have $|S_1| \ge e-d$. In particular, $x_0^{e-d} \mid $ LHS $- \prod \ell_i \implies x_0^{e-d} \mid \ell_0^{e} + \ell_{e+1}^e = \prod_{i} (\ell_0 - \zeta^{2i+1} \ell_{e+1})$, where $\zeta$ is $2e$-th root of unity. Since, $e-d \ge 2$ for $d \ge 4$, this simply implies that there are two indices $i_1$ and $i_2$ such that $\ell_0 - \zeta^{i_1} \ell_{e+1} = c_{i_1} x_0$, and $\ell_0 - \zeta^{i_2} \ell_{e+1} = c_{i_2} x_0$. Together, this implies that both $\ell_0$ and $\ell_{e+1}$ are multiples of $x_0$, which is a contradiction, since we assumed that $x_0$ cannot divide each $\ell_i$, for $i \in [0,e+1]$. Hence, we are done with case (a).

\medskip
\noindent {\bf Case (b):} Let $x_{0}^{e-d}\cdot P_{2,0}^{[d]} = \prod_{i \in [e]}\,\ell_i + \ell_0^{e-1} \cdot \ell_{e+1}$. We assume that $x_0$ does not divide both $\ell_i$, for some $i \in [e]$, and one of the $\ell_0$ or $\ell_{e+1}$, otherwise, we can divide by the maximum power $x_0$ both side. Again, a similar argument shows that $e \ge 2d-2$. 

Similarly, as before, we divide into subcases:
(b1) $x_0$ does not appear in $\ell_i$, for any $i \in [e]$,

(b2) $x_0$ appears in $\ell_i$, for some $i \in [e]$.

\medskip
\noindent {\bf Case (b1):} If $x_0$ does not appear in the first product, i.e,.~any of $\ell_i$, for $i \in [e]$, then it must appear in $\ell_{0}$ (because if it only appears in $\ell_{e+1}$, the degree of $x_0$ is $1$ in RHS, a contradiction). Note that, $x_0 \nmid \ell_{0}$ (and similarly $\ell_{e+1}$), because otherwise, substituting $x_0=0$ makes LHS $0$, while RHS remains $\prod_{i \in [e]} \ell_i$. Hence, let $\ell_0:= c_0 x_0 + \hat{\ell_0}$, where $\hat{\ell}_0$ is $x_0$-free. Substitute $x_0 = - \hat{\ell}_0/c_0$, so that 
\[ 
(- \hat{\ell}_0/c_0)^{e-d} \cdot P_{2,0}^{[d]} \;=\; \prod_{i \in [e]}\ell_i\;.
\] 
This in particular implies that $P_{2,0}^{[d]}$ is a product of linear forms, which is a contradiction by \cref{lem:irreducible-2}.

\medskip
\noindent {\bf Case (b2):} In this case, without loss of generality, $x_0$ appears in $\ell_1$. Note that, $x_0$ cannot divide $\ell_1$, because otherwise, it must divide LHS-$\prod_{i \in [e]} = \ell_{0}^{e-1}\ell_{e+1}$, which implies that $x_0$ must divide one of the $\ell_0$ or $\ell_{e+1}$, contradicting the minimality of $x_0$-division. Therefore, $\ell_1 = c_1 x_0 + \hat{\ell}_1$, where $c_1$ is a nonzero constant, and $\hat{\ell}_1$ is a nonzero linear form which is $x_0$-free. Substitute $x_0= - \hat{\ell}_1/c_1$, both side to get that
\[ 
(-\hat{\ell}_1/c_1)^{e-d}\cdot P_{2,0}^{[d]}\;=\;\hat{\ell}_0^{e-1}\hat{\ell}_{e+1}\;.
\] 
Therefore, again by unique factorization, we get that $f$ must a product of linear forms, which is a contradiction by \cref{lem:irreducible-2}.

\end{proof}

\section{Geometric complexity theory for product-plus-power}
\label{sec:orbitclosure}

In this section, we study computational and invariant theoretic properties of the polynomial 
\[
P^{[d]}_{r,s} = \sum_{k=1}^r \prod_{i=1}^d x_{ki} + \sum_{j=1}^s y_j^d,
\]
defined in \Cref{sec:GCTintro}; this is a polynomial of degree $d$ in $rd + s$ variables. The cases $(r,s) = (1,1), (1,2),(2,0),(0,r)$ correspond, respectively, to the product-plus-power, the product-plus-two-powers, the binomial, and the power sum polynomial mentioned in the previous sections. 

\Cref{thm:stab} determines the stabilizer of $P^{[d]}_{r,s}$ under the action of the group $\GL_{rd+s}$ acting on the variables. The knowledge of the stabilizer $H = \Stab_{\GL_{d+1}}(P_{1,1}^{[d]})$ allows us to determine the representation theoretic structure of the coordinate ring of the orbit of $P^{[d]}_{1,1}$, which is achieved in \Cref{pro:branchingrules}. Recall that the irreducible representations of $\GL_{d+1}$ are indexed by partitions $\la=(\la_1,\la_2,\ldots)$, $\la_1\geq\la_2\geq ...$, with $\ell(\la)\leq d+1$, see \ref{subsec:multiplicitiesorbit}.
Denote by $S_\la(\IC^{d+1})$ the irreducible representation of type $\la$.
In this section we write $S^d V = S_{(d)} V$ for the space of homogeneous polynomials of degree $d$ in $\dim V$ variables.
For every integer $D$ and every partition $\la$ of $dD$, we obtain the following identity:
\[
\mult_\la(\IC[\GL_{d+1}P^{[d]}_{1,1}])\;=\; 
\dim(S_\la \IC^{d+1})^H \;=\;
\sum_{\delta=0}^D \ \sum_{\mu\vdash \delta d, \mu\preceq\la , \ell(\mu)\leq d} \ a_\mu(d,\delta)\,,
\]
where $a_\mu(d,\delta)$ is the \emph{plethysm coefficients}, that is the multiplicity of $S_\mu(\IC^{d+1})$ in $\Sym^d(\Sym^\delta(V))$, see \Cref{pro:branchingrules}. We use this formula, and apply the \cite{IK20} approach to lower bounds on $\mult_\la(\IC[\overline{\GL_{d+1}P_{0,d}^{[d]}}])$,
to find a sequence of partitions where $\mult_\la(\IC[\GL_{d+1}P^{[d]}_{1,1}]) < \mult_\la(\IC[\overline{\GL_{d+1}P_{0,d}^{[d]}}])$, see \Cref{thm:obstruction}. This implies $ \overline{\GL_{d+1}P_{0,d}^{[d]}} \not\subseteq \overline{\GL_{d+1}P^{[d]}_{1,1}}$ if $d\geq 3$.

We  implement this approach explicitly and we determine via a computer calculation an abundance of multiplicity obstructions against generic polynomials, see \Cref{sec:calctables}.

In \Cref{pro:polystable}, we prove that $P^{[d]}_{r,s}$ is polystable, in the sense of invariant theory. This guarantees the existence of a \emph{fundamental invariant}, in the sense of \cite{BI17}: in \Cref{pro:alontarsi}, we show a connection between the degree of this fundamental invariant and the Alon-Tarsi conjecture on Latin squares in combinatorics.

\subsection{Stabilizer}\label{subsec:stabilizer}

Consider the action of the general linear group $\GL_n$ on the homogeneous components of the polynomial ring $\bbC[x_1 \vvirg x_n]$, by linear change of variables, as described in \Cref{sec:GCTintro}. For a homogeneous polynomial $f \in \bbC[\bfx]_d$, write $\Stab_{\GL_n}(f)$ for its stabilizer under this action. It is an immediate fact that $\Stab_{\GL_n}(f)$ is a closed algebraic subgroup of $\GL_n$. It may consists of several connected (irreducible) components: the \emph{identity component}, denoted $\Stab^0_{\GL_n}(f)$ is the connected component containing the identity; $\Stab^0_{\GL_n}(f)$ is a closed, normal subgroup of $\Stab_{\GL_n}(f)$ \cite[Lemma 2.1]{Ges:Geometry_of_IMM}; the quotient $\Stab_{\GL_n}(f) / \Stab^0_{\GL_n}(f)$ is a finite group.

The Lie algebra $\frakg$ of an algebraic group $G$ can be geometrically identified with the tangent space to $G$ at the identity element. Moreover, if $G$ is a subgroup of $\GL_n$, then $\frakg$ is naturally a subalgebra of $\frakgl_n = \End(\bbC^n)$; moreover $\frakg$ uniquely determined the identity component of $G$.

It is a classical fact that the Lie algebra of $\Stab_{\GL_n}(f)$ is the annihilator of $f$ under the Lie algebra action of $\frakgl_n$ on $\bbC[\bfx]_d$; denote this annihilator by $\frakann_{\frakgl_n}(f)$. Typically, in order to determine $\Stab_{\GL_n}(f)$, one first computes $\frakann_{\frakgl_n}(f)$, which uniquely determines $\Stab_{\GL_n}^0(f)$. Then, one determines $\Stab_{\GL_n}(f)$ as a subgroup of the normalizer $N_{\GL_n} \Stab_{\GL_n}^0(f)$. The last step is often challenging; some methods to do this systematically in simplified settings are presented in \cite{GarGur:simple_poly_stab,Ges:Geometry_of_IMM,GesIkPa:GCTMatrixPowering}.

First, we record a general result regarding the stabilizer of sums of polynomials in disjoint sets of variables. This is the symmetric version of \cite[Thm. 4.1(i)]{ConGesLanVenWan:GeometryStrassenAsyRankConj}. We say that a polynomial $f \in S^d \bbC^n$ is \emph{concise} (in $S^d \bbC^n$) if the first order partials of $f$ are linearly independent.

\begin{lemma}\label{lemma: direct sum stabilizer algebras}
Let $V = V_1 \oplus V_2$ and let $f \in \bbC[V^*]_d = S^d V$ be a homogeneous polynomial with $f = f_1 + f_2$, where $f_i \in S^d V_i$ are both concise, with $d \geq 3$. Then 
\begin{enumerate}
 \item[(i)] $\frakann_{\frakgl(V)}(f_1) = \frakann_{\frakgl(V_1)}(f_1) \oplus \Hom(V_2,V)$;
 \item[(ii)] $\frakann_{\frakgl(V)}(f_1 + f_2) = \frakann_{\frakgl(V_1)}(f_1) \oplus \frakann_{\frakgl(V_2)}(f_2)$.
\end{enumerate}
\end{lemma}
\begin{proof}
For both statements, the inclusion of the right-hand term into the left-hand term is clear. We prove the reverse inclusion.

For $X \in \frakgl(V)$, write $X = \sum_{i,j=1}^2 X_{ij}$, with $X_{ij} \in \Hom(V_i,V_j)$. 
 
 The proof of (i) amounts to showing that if $X \in \frakann_{\frakgl(V)}(f_1)$, then $X_{12} = 0$ and $X_{11} \in \frakann_{\frakgl(V_1)}(f_1)$. Suppose $X.f_1 = 0$. Notice $X.f_1 = X_{11}.f_1 + X_{12}.f_1$; here $X_{11}.f_1 \in S^d V_1$ and $X_{12}.f_1 \in V_2 \otimes S^{d-1}V_1$. In particular, both terms must vanish. The term $X_{12}.f_1$ is a sum of at most $\dim V_2$ linearly independent elements, each of which is a linear combination of first order partials of $f_1$. Therefore, all such linear combinations must be $0$, and since $f_1$ is concise in $S^d V_1$, we obtain $X_{12} = 0$. The condition $X_{11}.f_1  =0$ is, by definition, equivalent to $X_{11}\in \frakann_{\frakgl(V_1)}(f_1)$. This conclude the proof of (i).
 
 To prove (ii), we show that if $X \in \frakann_{\frakgl(V)}(f)$, then $X_{12}= 0$, $X_{21} = 0$ and $X_{ii} \in \frakann_{\frakgl(V_i)}(f_i)$. Suppose $X.f= 0$. We have $X.f = (X_{11} + X_{12}).f_1 + (X_{21}+X_{22}).f_2$. Notice 
 \begin{align*}
\begin{array}{lll}  
   X_{11} f_1 \in S^d V_1, & ~~~~~~~ &    X_{12} f_1 \in V_2 \otimes S^{d-1} V_1, \\
   X_{21} f_2 \in V_1 \otimes S^{d-1} V_2, & ~~~~~~~ &    X_{22} f_2 \in S^d V_2 .
   \end{array}
 \end{align*}
 Since $d \geq 3$, the four terms are linearly independent, hence they all must vanish individually. This concludes the proof.
\end{proof}

We can now determine the stabilizer of $P^{[d]}_{r,s}$. Let $\bbT^{\SL_n}$ denote the subgroup of diagonal elements in $\SL_n$. We use the \emph{wreath product} notation: given a group $G$, the wreath product $G \wreath \frakS_k$ for the semidirect product $G^{\times k} \rtimes \frakS_k$ where $\frakS_k$ acts on the direct product $G^{\times k}$ by permuting the direct factors; we refer to \cite[Sec. 1.6]{RobinsonGroupTheory} for details on this construction. 
\begin{theorem}\label{thm:stab}
 For $d \geq 3$ and for every $r,s$, we have
 \[
 \Stab_{\GL(V)}( P^{[d]}_{r,s} ) = ([\bbT^{\SL_d} \rtimes \frakS_{d}] \wreath \frakS_r ) \times (\bbZ_d \wreath \frakS_s);
 \]
each copy of $\bbT^{\SL_d} \rtimes \frakS_{d}$ acts by rescaling and permuting the variables in one of the $r$ sets $\{ x_{ji} : i = 1 \vvirg d\}$ for $j=1 \vvirg r$; the group $\frakS_r$ permutes (set-wise) these sets; the group $\bbZ_d \wreath \frakS_s$ acts by rescaling (by a $d$-th root of $1$) and permuting the variables in the set $\{ y_i : i = 1 \vvirg s\}$.
\end{theorem}
\begin{proof}
It is clear that the group on the right-hand side is contained in the stabilizer $\Stab_{\GL(V)}( P^{[d]}_{r,s} )$. We show the reverse inclusion.

First, we determine the identity component of $\Stab_{\GL(V)}( P^{[d]}_{r,s} )$. By \Cref{lemma: direct sum stabilizer algebras}, the annihilator of $P^{[d]}_{r,s}$ in $\frakgl(V)$ is the direct sum of the annihilators of its summands.  This guarantees that the identity component of the stabilizer is the product of the identity components for the summands of $P^{[d]}_{r,s}$. The identity component for each square-free monomial is a copy of $\bbT^{\SL_d}$, see, e.g., \cite[Sec. 7.1.2]{Lan:GeometryComplThBook}. The identity component of each power is trivial. Therefore, we deduce $\Stab_{\GL(V)}^0(P^{[d]}_{r,s}) = ({\bbT^{\SL_d}})^{\times r}$.
 
 Since $\Stab_{\GL(V)}^0(P^{[d]}_{r,s})$ is a normal subgroup of $\Stab_{\GL(V)}(P^{[d]}_{r,s})$, we have 
 \[
  \Stab_{\GL(V)}( P^{[d]}_{r,s} ) \subseteq N_{\GL_{rd+s}}({\bbT^{\SL_d}}^{\times r}) = ([\bbT^{\SL_d} \rtimes \frakS_d] \wreath \frakS_r ) \rtimes Q
 \]
where $Q$ is the parabolic subgroup stabilizing the subspace spanned by the $x_{ij}$ variables; here $N_{\GL_{rd+s}} (=)$ denotes the normalizer subgroup.
 
In order to determine the discrete component, we follow the same argument as the one used for the power sum polynomial $P^{[d]}_{0,s}$ in \cite[Section 8.12.1]{Lan:GeometryComplThBook}. In particular, $ \Stab_{\GL(V)}( P^{[d]}_{r,s} )$ stabilizes the Hessian determinant of $P^{[d]}_{r,s}$, up to scaling. A direct calculation shows that this Hessian determinant, up to scaling, is 
\[
 H = (\prod_{i,j}x_{ij} \prod_k y_k)^{d-2}.
\]
Unique factorization implies that $\Stab_{\GL(V)}( P^{[d]}_{r,s} ) \cap Q \subseteq \bbT \rtimes \frakS_s$, where $\bbT$ is the torus of diagonal matrices acting on the $y_j$ variables. Hence this subgroup commutes with $[\bbT^{\SL_d} \rtimes \frakS_d] \wreath \frakS_r $ and we deduce 
\[
 \Stab_{\GL(V)}( P^{[d]}_{r,s} ) \cap Q = \Stab_{\GL_s}(y_1^d + \cdots + y_s^d) = \bbZ_d \wreath \frakS_s.
\]
This concludes the proof.
\end{proof}

In the context of geometric complexity theory it is important to know if the polynomial is characterized by its stabilizer~\cite{MS08}.
While this property fails for the polynomials $P^{[d]}_{r, s}$, a slightly weaker statement is true --- every polynomial stabilized by $\Stab(P^{[d]}_{r, s})$ is in the orbit of $P^{[d]}_{r, s}$ or a very special restriction of $P^{[d]}_{r, s}$.
This is similar to the properties of minrank tensors and slice rank tensors considered in~\cite{https://doi.org/10.48550/arxiv.1911.02534}.

\begin{theorem}
\label{thm:characterized}
    If a polynomial $f \in \bbC[x_{11}, \dots, x_{dr}, y_1, \dots, y_s]_d$ is stabilized by $\Stab(P^{[d]}_{r, s})$, then
    \[
    f = \alpha \sum_{i=1}^r \textprod_{j=1}^{d} x_{ji} + \beta \sum_{i=1}^s y_{i}^d,
    \]
    for some $\alpha, \beta \in \bbC$.
\end{theorem}
\begin{proof}
    Partition the set of variables into the subsets $X_{i} = \{x_{1i}, \dots, x_{di}\}$ and $Y_i = \{y_i\}$.
    Note that $\Stab(P^{[d]}_{r, s})$ contains the transformation which scales all variables in one of the subsets by a $d$-th root of unity, acting as identity on all other variables.
    It follows that each monomial of $f$ contains variables from only one of the subsets, because if this was not the case then the transformation described above multiplies the monomial by a coefficient different from $1$.
    Thus we have
    \[
    f = \sum_{i = 1}^r f_i(x_{1i}, \dots, x_{di}) + \sum_{i = 1}^s \beta_i y_i^d.
    \]

    Since $f$ is fixed under the symmetric group $\frakS_s$ permuting $y_1, \dots, y_s$, the coefficients $\beta_i$ are all equal.
    Since $f$ is fixed under the symmetric group $\frakS_r$ permuting the subsets $X_1, \dots, X_r$, all the polynomials $f_i$ also coincide.

    Finally, the stabilizer group contains the transformations scaling $x_{j1}$ by $\lambda$ and $x_{k1}$ by $\lambda^{-1}$.
    This transformation scales a monomial $x_{j1}^{p_j} x_{k1}^{p_k} \dots$ by $\lambda^{p_j - p_k}$. It follows that each monomial of $f_1$ must have the same degree with respect to each variable, that is, $f_1 = \alpha \prod_{j = 1}^d x_{j1}$.
\end{proof}

\subsection{Multiplicities in the coordinate ring of the orbit}
\label{subsec:multiplicitiesorbit}

A \emph{partition} $\la = (\la_1,\la_2,\ldots)$ is a finite non-increasing sequence of nonnegative integers. We write $\ell(\la) := \max\{i\mid\la_i\neq 0\}$, and $\la\vdash D$ means $\sum_i \la_i = D$.
To each partition $\la$ we associate its Young diagram, which is a top-left justified array of boxes with $\la_i$ boxes in row $i$. For example, the Young diagram of $\la=(4,4,3)$ is
$\young(\ \ \ \ ,\ \ \ \ ,\ \ \ )$.
The transpose of the Young diagram is obtained by switching rows and columns.
Denote the partition corresponding to this Young diagram by $\la^t$, for example $(4,4,3)^t = (3,3,3,2)$.
A group homomorphism $\varrho:\GL_D \to \GL(V)$, where $V$ is a finite dimensional complex vector space, is called a \emph{representation} of $\GL_D$. A representation is \emph{polynomial} if each entry of the matrix corresponding to the linear map $\varrho(g)$ is given by a polynomial in the entries of the elements of $\GL_D$.
A linear subspace that is closed under the group operation is called a \emph{subrepresentation}.
A representation with only the two trivial subrepresentations is called \emph{irreducible}.
The irreducible polynomial representations of $\GL_{d+1}$ are indexed by partitions $\la$ with $\ell(\la)\leq d+1$, see for example~\cite[Ch.~8]{Ful97}. Denote by $S_\la(\IC^{d+1})$ the irreducible representation of type $\la$.
For a $\GL_{d+1}$-representation $V$, write $\mult_\la(V)$ to denote the multiplicity of $\la$ in $V$, i.e., the dimension of the space of equivariant maps from $S_\la(\IC^{d+1})$ to $V$, or equivalently, the number of summands of isomorphism type $\la$ in any decomposition of $V$ into a direct sum of irreducible representations.

In this section, we consider the representations given by the homogeneous components of the coordinate ring of $\bar{\GL_{d+1} \cdot P_{1,1}^{[d]}}$. By \Cref{thm:stab}, the stabilizer of $P_{1,1}^{[d]}$ under the action of $\GL_{d+1}$ is $H := \Stab_{\GL_{d+1}}(P^{[d]}_{1,1}) \simeq \bbZ_d  \times (\bbT^{\SL_d} \rtimes \frakS_d)$.

The stabilizer is used to determine the multiplicities in the coordinate ring of the group orbit $\mult_\la(\IC[\GL_{d+1} P^{[d]}_{1,1}])$. It is a general fact that $\IC[\GL_{d+1} P^{[d]}_{1,1}] \simeq \IC[\GL_{d+1}]^H$: in other words the coordinate ring of the orbit coincides with the subring of $H$-invariant elements in the coordinate ring of the group $\GL_{d+1}$. This ring of $H$-invariants can be determined using the Algebraic Peter-Weyl Theorem \cite[Thm. 4.2.7]{GooWal:Symmetry_reps_invs}, a powerful tool that allows us to compute the relevant multiplicities. In turn, we have  
\[
\mult_\la(\IC[\GL_{d+1} P^{[d]}_{1,1}]) = \dim (S_\la \IC^{d+1})^{H}.
\]
We determine the dimension of these invariant spaces by classical representation branching rules, see \Cref{pro:branchingrules}.

For partitions $\mu$ and $\la$, define $\mu\preceq\la$ if and only if the Young diagram of $\mu$ is contained in the one of $\la$,  (that is for every $i$, one has $\mu_i\leq \la_i$) and the skew diagram $\la/\mu$ given by the difference has at most 1 box in each column (that is $\la^t_i-\mu^t_i\leq 1$). In other words, $\lambda$ can be obtained from $\mu$ by adding a suitable number of boxes, with no two of them in the same column.
The plethysm coefficient is defined to be the multiplicity $a_\mu(d,D) := \mult_\mu(S^d(S^D(\bbC^N)))$, and it does not depend on $N$ as long as $N \geq d$.
\begin{proposition}\label{pro:branchingrules}
For $\la\vdash dD$ we have
\[\mult_\la(\IC[\GL_{d+1}P^{[d]}_{1,1}]) = \dim(S_\la \IC^{d+1})^H =
\sum_{\delta=0}^D\sum_{\substack{\mu\vdash \delta d \\ \mu\preceq\la \\\ell(\mu)\leq d}} a_\mu(d,\delta).
\]
\end{proposition}
\begin{proof}
\[
(S_\la \IC^{d+1})^H
=
(S_\la (\bbC\oplus\bbC^d)\downarrow^{\GL_{d+1}}_{\GL_1\times\GL_d})^{\bbZ_d  \times (\bbT^{\SL_d}  \rtimes \frakS_d)}
\stackrel{\text{Pieri's rule}}{=}
\bigoplus_{\substack{\mu\preceq\la\\ \ell(\mu)\leq d}}(S^{|\la|-|\mu|}\bbC^1)^{\bbZ_d}\otimes (S_\mu\bbC^d)^{\bbT^{\SL_d} \rtimes \frakS_d},
\]
where Pieri's rule is a well-known decomposition rule, see for example \cite[p.~80, Exe.~6.12]{FH91}.
Now, $\dim((S^{|\la|-|\mu|}\bbC^1)^{\bbZ_d}) = 1$ if and only if $|\la|-|\mu|$ is a multiple of $d$ if and only if $|\mu|$ is a multiple of $d$.
Otherwise it is~0.
Hence
\[
\dim (S_\la V)^H
=
\sum_{\delta=0}^d\sum_{\substack{\mu\vdash\delta d\\\mu\preceq\la\\\ell(\mu)\leq d}} \underbrace{\dim (S_\mu\bbC^d)^{\bbT^{\SL_d} \rtimes \frakS_d}}_{=a_{\mu}(d,\delta)}
\]
The last underbrace equality is Gay's theorem \cite{Gay76}.
\end{proof}
The condition that $\ell(\mu)\leq d$ is not necessary, because if $\ell(\mu)>d$, then $a_\mu(d,\delta)=0$.

A computer calculation shows that this indeed gives multiplicity obstructions in the sense of \Cref{sec:GCTintro}. We record the result of the this calculation in \Cref{sec:calctables}. We used the HWV software \cite{BHIM22} to directly calculate that (10, 6, 4, 4) and (8, 8, 4, 4) are the only types in the vanishing ideal for $D=8$, $d=3$. For $d=3$ there are no equations in degree $1,\ldots,7$. In particular, none of Brill's equations from \cite{Gor94}, which are of degree $d+1$, vanishes on $\GL_{d+1} P^{[d]}_{1,1} \cap S^d\IC^d$. 

\subsection{Polystability}

A polynomial $f \in S^d V$ is called \emph{polystable} if its $\SL(V)$-orbit is closed.
Polystability is an important property in GCT, as it implies the existence of a \emph{fundamental invariant} that connects the $\GL$-orbit with the $\GL$-orbit closure, see \cite[Def.~3.9 and Prop.~3.10]{BI17}. This connection can be used to exhibit multiplicity obstructions, as was done in \cite{IK20}.

\begin{proposition}\label{pro:polystable}
Let $d\geq 2$. The polynomial $P^{[d]}_{r,s}$ is polystable, i.e., the orbit $\SL_{rd+s} P^{[d]}_{r,s}$ is closed.
\end{proposition}
\begin{proof}
If $d = 2$, then $P^{[2]}_{r,s}$ is a polynomial of degree $2$ defining a quadratic form of maximal rank. This is polystable.

Suppose $d \geq 3$. A criterion for polystability is given in \cite[Prop. 2.8]{BI17}, based on works of Hilbert, Mumford, Luna, and Kempf. 

In order to apply this criterion, consider the group $R =  \Stab(P^{[d]}_{r,s}) \cap \bbT$, where $\bbT$ denotes the torus of diagonal matrices in $\GL_{rd+s}$, in the basis defined by the variables. By \Cref{thm:stab}, we deduce $R =(\bbT^{\SL_d})^{\times r} \times \bbZ_d^{\times s}$. This is a group consisting entirely of diagonal matrices and it is easy to verify that its centralizer in $\SL_{d+1}$ coincides with $\bbT^{\SL_{d+1}}$. This proves the first property of the criterion.

For the second property, consider the exponent vectors of the monomials appearing in $P^{[d]}_{r,s}$. For a monomial $m$, write $\wt(m)$ for its exponent vector. It is immediate to verify that 
\[
\sum_{i=1}^r \wt(x_{i1}\cdots x_{id}) + \frac{1}{d}\sum_{j=1}^s \wt(y_j^d) = (1 \vvirg 1);
\]
this shows that the vector $(1 \vvirg 1)$ lies in the convex cone generated by the exponent vectors of the monomials of $P^{[d]}_{r,s}$. This proves the second part of the criterion and concludes the proof.
\end{proof}

\cref{pro:polystable} reduces to the following in the special case $r=s=1$:
\begin{corollary}
Let $d\geq 2$. The product-plus-power polynomial $P^{[d]}_{1,1}$ is polystable, i.e., the orbit $\SL_{d+1}P^{[d]}_{1,1}$ is closed.
\end{corollary}

\subsection{Fundamental invariants and the Alon-Tarsi conjecture}
\label{sec:fundinvalontarsi}

The \emph{fundamental invariant} $\Phi$ of a polystable polynomial $f \in S^D V$ is the unique (up to scaling) smallest degree $\SL(V)$-invariant function in $\IC[\overline{\GL(V)f]}$, see Def.~3.8 in \cite{BI17}.
It describes the connection between the orbit and the orbit-closure of $f$: more formally, the coordinate ring of the orbit $\IC[\GL(V)f]$ is canonically isomorphic to the localization at $\Phi$ of the coordinate ring of the orbit-closure, that is $\IC[\overline{\GL(V)f}]_\Phi$; see \cite[Pro.~3.9]{BI17}.
This connection can be used to exhibit multiplicity obstructions, as was done in \cite{IK20}.

It is known that for even $d$ the orbit closure
$\overline{\GL_{d}(x_1 \cdots x_d)}$ of a squarefree monomial
has fundamental invariant of degree $d$ if and only if the Alon-Tarsi conjecture on Latin squares holds for $d$; see \cite{KumLan:ConnectionsAlonTarsi_HadHowe} and \cite[Prop.~3.26]{BI17}; otherwise the fundamental invariant has higher degree. In this section we show an analogous result for the orbit closure $\overline{\GL_{d+1}(x_1 \cdots x_d+x_{d+1}^d)}$: if $d$ is even this orbit closure has fundamental invariant of degree $d+1$ if and only if the Alon-Tarsi conjecture on Latin squares holds for $d$; otherwise the fundamental invariant has higher degree.

\begin{proposition}
\label{pro:alontarsi}
Let $d$ be even.
The degree of the fundamental invariant of $P^{[d]}_{1,1}$ is $d+1$ if and only if the Alon-Tarsi conjecture for $d$ is true, otherwise it is of higher degree.
\end{proposition}
\begin{proof}
We follow the presentation in \cite{CIM17, BI17, blaser2020complexity}.
For a partition $\la$ we place positive integers into the boxes of the Young diagram and call it a \emph{tableau} $T$ of shape $\la$.
The vector of numbers of occurrences of 1s, 2s, etc, is called the \emph{content} of $T$.
The content is $n \times d$ if $T$ has exactly $d$ many 1s, $d$ many 2s, $\ldots$, $d$ many $n$s.
The set of boxes of the Young diagram of $\la$ is denoted by $\textup{boxes}(\la)$.
The boxes that have the same number are said to form a \emph{block}.

Let $m=n+1$.
Fix a tableau $T$ of shape $\la$ with content $n \times d$ and fix a tensor $p = \sum_{i=1}^r \ell_{i,1}\otimes\cdots\otimes\ell_{i,d} \in \otimes^d \IC^m$.
A placement
\[
\vartheta : \textup{boxes}(\la) \to [r] \times [d]
\]
is called \emph{proper} if
the first coordinate of $\vartheta$ is constant in each block
and the second coordinate of $\vartheta$ in each block is a permutation.
We define the determinant of a matrix that has more rows than columns as the determinant of its largest top square submatrix.

For a tableau $T$ with content $\Delta\times d$
we define the polynomial $f_T$ via its evaluation on $p$:
\begin{equation}\label{eq:sumpropertheta}
    f_T(p)
    := \sum_{\text{proper\ }\vartheta} \ \prod_{c = 1}^{\lambda_1}\det{}_{\vartheta, c} \ \text{ with } \ \det{}_{\vartheta, c} := \det\left(\ell_{\vartheta(1, c)} \dots \ell_{\vartheta(\mu_c, c)}\right)
\end{equation}

The degree of $f_T$ is $\Delta$.
The polynomial $f_T$ is $\SL_m$-invariant if and only if the shape of $T$ is rectangular with exactly $m$ many rows.
It is easy to see that $f_T=0$ if $T$ has any column in which a number appears more than once.
Moreover, it is easy to see that $f_T$ is fixed (up to sign) when two entries in $T$ are exchanged within a column.
So, up to sign, there is only one $T$ that could give an $\SL_m$-invariant of degree $d+1$: It is the tableau with $m=d+1$ many rows and $d$ columns that has only entries $i$ in row $i$. For $n=4$ it looks as follows.

\[
T = \young(1111,2222,3333,4444,5555)\]

For this $T$ it remains to verify that $f_T$ does not vanish identically on to orbit closure $\overline{\GL_{d+1}(x_1\cdots x_d+x_{d+1}^d)}$. Since $f_T$ is $\SL_{d+1}$-invariant, this is equivalent to $f_T$ not vanishing at the point $x_1\cdots x_d+x_{d+1}^d$.
So we now evaluate $f_T(x_1\cdots x_d+x_{d+1}^d)$.
The nonzero summands in \cref{eq:sumpropertheta} must place $(d+1,\ast)$ into one of the blocks.
We can partition the summands according to the row in which $(d+1,\ast)$ is placed.
Since the number of columns is even,
each part of the partition contributes the same number to the overall sum. That number is the column sign of the unique Latin square that is obtained when removing the row in which $(d+1,\ast)$ is placed.
Hence the whole sum if $d+1$ times the difference of the column-even and column-odd Latin squares, so its nonvanishing is equivalent to the Alon-Tarsi conjecture for $d$.
\end{proof}

\begin{remark}
Other fundamental invariants connected to the Alon-Tarsi conjecture have recently been studied in \cite{LZX21,10.1093/imrn/rnac311}.
\end{remark}

\subsection{New obstructions}
\label{subsec:obstructions}
For two partitions $\la$ and $\mu$, their sum is defined coordinatewise, i.e., $(\la+\mu)_i = \la_i+\mu_i$. We write $a \times b$ for the partition $(b,b,\ldots,b)$ of $ab$.
For example, if $d=3$, then the Young diagram to $\la := (5d-1,1)+((d+1)\times (10d))$ is the following:\\
{\tiny
\yng(44,31,30,30)
}\\

This section is devoted to proving the following result, which is a restatement of \Cref{thm:intro:obstructions}.
\begin{theorem}\label{thm:obstruction}
Let $d \geq 3$ be even, and let $\la := (5d-1,1)+\big((d+1)\times (10d)\big)$. Then we have representation theoretic multiplicity obstructions:
\[
\mult_\la(\overline{\IC[\GL_{d+1} P^{[d]}_{1,1}]}) \leq 4 < 5 = \mult_\la(\overline{\IC[\GL_{d+1} (x_1^d+\cdots +x_{d+1}^d)]}),\]
and hence
$\overline{\GL_{d+1} (x_1^d+\cdots +x_{d+1}^d)} \not\subseteq \overline{\GL_{d+1} P^{[d]}_{1,1}}$.
\end{theorem}
We point out that these obstructions of \Cref{thm:obstruction} are only based on the symmetries of the two polynomials as in \cite{IK20}.

The upper and the lower bound are proved independently, see \Cref{pro:boundFOUR} and \Cref{pro:boundFIVE}, which proves the theorem. Fix the following notation: $\kappa := (5d-1,1)$,  $\blacksquare:=(d+1)\times (10d)$,
$\square:=d\times (10d)$,
$\Delta := |\square|/d = 10d$,
and 
$\la := \kappa+\blacksquare$.

\begin{proposition}
\label{pro:boundFOUR}
$\mult_\la(\IC[\overline{\GL_{d+1}P^{[d]}_{1,1}}]) \leq \mult_\la(\IC[\GL_{d+1}P^{[d]}_{1,1}]) = \mult_\kappa(\IC[\GL_{d+1}P^{[d]}_{1,1}]) = 4$.
\end{proposition}
\begin{proof}
The ring $\IC[\overline{\GL_{d+1}P^{[d]}_{1,1}}]$ is 
a localization of the ring $\IC[\GL_{d+1}P^{[d]}_{1,1}]$, see \cite{BI17}, which implies $\mult_\la(\IC[\overline{\GL_{d+1}P^{[d]}_{1,1}}]) \leq \mult_\la(\IC[\GL_{d+1}P^{[d]}_{1,1}])$.
We observe that $a_{\nu+\square}(d,i+\Delta) = a_{\nu}(d,i)$, because $\square$ has an even number of columns and exactly $d$ rows.
Then we calculate:
\begin{eqnarray*}
\mult_\la(\IC[\GL_{d+1}P^{[d]}_{1,1}])
&=&
\sum_{\delta=0}^5\sum_{\substack{\mu\vdash \delta d \\ \mu\preceq\la \\\ell(\mu)\leq d}} a_\mu(d,\delta)
\\
&=&
 a_{(d)+\square}(d,1+\Delta)+a_{(d-1,1)+\square}(d,1+\Delta)+
a_{(2d)+\square}(d,2+\Delta)
 \\
&&+a_{(2d-1,1)+\square}(d,2+\Delta)+a_{(3d)+\square}(d,3+\Delta)+
a_{(3d-1,1)+\square}(d,3+\Delta)
\\
&&+
a_{(4d)+\square}(d,4+\Delta)+a_{(4d-1,1)+\square}(d,4+\Delta)+
a_{(5d-1,1)+\square}(d,5+\Delta)
\\
&=&
 a_{(d)}(d,1)+a_{(d-1,1)}(d,1)+
a_{(2d)}(d,2)+a_{(2d-1,1)}(d,2)+a_{(3d)}(d,3) \\
&&+
a_{(3d-1,1)}(d,3)+
a_{(4d)}(d,4)+a_{(4d-1,1)}(d,4)+
a_{(5d-1,1)}(d,5)
\\
&=& 4,
\end{eqnarray*}
because $a_{(nm)}(n,m)=1$, and $a_{(nm-1,1)}(n,m)=0$, because $(nm-1,1)$ is of hook shape, see \cite[Prop.~19.3.20]{BI18}. Note that there is no summand $a_{(5d)}(d,5)$ and no summand $a_{(0)}(d,0)$,
because $(5d) \not\preceq (5d-1,1)$, and $(0) \not\preceq (5d-1,1)$.
\end{proof}

\begin{proposition}
\label{pro:boundFIVE}
$\mult_\la(\IC[\overline{\GL_{d+1}(x_1^d+\cdots+x_{d+1}^d)}]) \geq 5$.
\end{proposition}
\begin{proof}
We use the Main Technical Theorem 4.2 from \cite{IK20}.
Consider all partitions $\varrho$ of 5, and observe that $\sum_{i=1}^{d+1} 2\lceil \frac{\varrho_i}{2(d-2)}\rceil \leq 10$.
In the notation of \cite{IK20}, we set $e_\Xi := 10$, which is exactly how many $(d+1)\times d$ blocks form $\blacksquare$.

For a partition $\varrho \vdash_m D$ the \emph{frequency notation} $\hat\varrho \in \IN^m$ is defined via
$\hat\varrho_i := |\{j \mid \varrho_j = i\}|$.
For example, the frequency notation of $\varrho=(3,3,2,0)$ is $\hat\varrho=(0,1,2,0)$.
We observe that $|\varrho|=\sum_i i \hat\varrho_i$.
We first use Theorem 4.1 from \cite{IK20} (with adjusted notation):

Let $m := d+1$, $D:=5$, $\kappa = (5d-1,1) \vdash_{m} Dd$.
Define 
\[
b(\kappa,\varrho,d,D) := \sum_{\mu^1,\mu^2,\ldots,\mu^D \atop \mu^i \vdash d i \hat\varrho_i} c_{\mu^1,\mu^2,\ldots,\mu^D}^\kappa \prod_{i=1}^D a_{\mu^i}(\hat\varrho_i,i\cdot d).
\]
Then
\[
\mult_{\kappa} \IC[\GL_m (x_1^d + x_2^d + \cdots + x_m^d)] = \sum_{\varrho\vdash_m D} b(\kappa,\varrho,d,D).
\]
For the multi-Littlewood-Richardson coefficient to be nonzero, it is necessary that all $\mu^i \subseteq (5d-1,1)$, so each $\mu^i$ is either a single row or a hook $(\widehat\varrho_i\cdot i \cdot d -1 ,1)$.
But $a_{nm-1,1}(n,m) = 0$ and $a_{nm}(n,m) = 1$, so we can assume that the sum has only the summand with $\mu^i = (\widehat\varrho_i\cdot i \cdot d)$ and the product of plethysm coefficients is 1.
Hence, the multi-Littlewood-Richardson coefficient counts the number of semistandard tableaux of shape $(5d-1,1)$ and content $(\mu^1,\ldots,\mu^5)$.

It is instructive to look at all possible $\widehat\varrho$:
$\widehat{(1,1,1,1,1)}=(5)$,
$\widehat{(2,1,1,1)}=(3,1)$,
$\widehat{(2,2,1)}=(1,2)$,
$\widehat{(3,1,1)}=(2,0,1)$,
$\widehat{(3,2)}=(0,1,1)$,
$\widehat{(4,1)}=(1,0,0,1)$,
$\widehat{(5)}=(0,0,0,0,1)$.
We observe that $\widehat\varrho$ has exactly two nonzero entries in 5 cases, and only one nonzero entry in 2 cases. There are no semistandard tableaux of shape $(5d-1,1)$ with only one entry, and there is exactly one semistandard tableaux of shape $(5d-1,1)$ with two symbols and fixed content.
Hence
\[
\mult_{\kappa} \IC[\GL_{d+1} (x_1^d + x_2^d + \cdots + x_{d+1}^d)] = 5.
\]
Note that this argument works indeed for all $d\geq 3$, even though for $d=3$ we do not have $\varrho=(1,1,1,1,1)$ in the sum (because it has more than $d+1=4$ rows, but its contribution is zero anyway).

We now apply Theorem 4.2 from \cite{IK20}, which implies
\[
\mult_{\blacksquare + \kappa} \IC[\overline{\GL_{d+1} (x_1^d + x_2^d + \cdots + x_{d+1}^d)}] \geq 5.\qedhere
\]
\end{proof}

\subsubsection{Acknowledgements} P.D.\ is supported by the project titled~``Computational Hardness of Lattice Problems and Implications", funded by National Research Foundation (NRF) Singapore. The work of F.G.\ is partially supported by  the Thematic Research Programme ``Tensors: geometry, complexity and quantum entanglement'', University of Warsaw, Excellence Initiative -- Research University and the Simons Foundation Award No. 663281 granted to the Institute of Mathematics of the Polish Academy of Sciences for the years 2021-2023.
C.I.\ was supported
by the DFG grant IK 116/2-1 and the EPSRC grant EP/W014882/1.
Part of the work was done while V.L. was affiliated with the QMATH Centre, University of Copenhagen.
V.L. acknowledges financial support from VILLUM FONDEN via the QMATH Centre of Excellence (Grant No. 10059)
and the European Union (ERC Grant Agreements 818761 and 101040907). 
Views and opinions expressed are however those of the author(s) only and do not necessarily reflect those of the European Union or the European Research Council Executive Agency. Neither the European Union nor the granting authority can be held responsible for them.

{\small
\newcommand{\etalchar}[1]{$^{#1}$}

}

\newpage
\appendix

\section{Calculation tables}
\label{sec:calctables}

We list the partitions $\la$ for which the plethysm coefficient $a := a_\la(\delta,d)$ exceeds the multiplicity $b := \mult_\la(\IC[\GL_{d+1}(x_1\cdots x_d+x_{d+1}^d)])$. We write $\la_{a > b}$.
We list $\la$ always with all $d+1$ parts, i.e., with all trailing zeros. $\la$ always has $d\delta$ many boxes.
If we list a case $(d,\delta)$ and not list $(d,\delta')$ with $\delta'<\delta$, then this means that $(d,\delta')$ is empty.

\subsection*{$d=3$, $\delta=8$:}
$(8,8,4,4)_{2>1}$, 
$(10,6,4,4)_{4>3}$

\subsection*{$d=4$, $\delta=6$:}
$(6,6,4,4,4)_{1>0}$, 
$(7,7,5,5,0)_{1>0}$, 
$(7,7,7,3,0)_{1>0}$, 
$(8,5,5,3,3)_{1>0}$

\subsection*{$d=4$, $\delta=7$:}
$(7,7,5,5,4)_{1>0}$, 
$(7,7,6,5,3)_{1>0}$, 
$(7,7,7,4,3)_{1>0}$, 
$(7,7,7,5,2)_{1>0}$, 
$(7,7,7,7,0)_{1>0}$, 
$(8,6,6,4,4)_{4>1}$, 
$(8,7,5,4,4)_{1>0}$, 
$(8,7,5,5,3)_{2>0}$, 
$(8,7,6,4,3)_{4>2}$, 
$(8,7,6,5,2)_{4>1}$, 
$(8,7,7,3,3)_{3>0}$, 
$(8,7,7,4,2)_{1>0}$, 
$(8,7,7,5,1)_{3>0}$, 
$(8,8,4,4,4)_{4>2}$, 
$(8,8,5,4,3)_{4>1}$, 
$(8,8,6,4,2)_{9>4}$, 
$(8,8,7,3,2)_{3>1}$, 
$(8,8,7,4,1)_{4>3}$, 
$(8,8,8,2,2)_{3>2}$, 
$(9,6,5,4,4)_{3>0}$, 
$(9,6,5,5,3)_{1>0}$, 
$(9,6,6,4,3)_{5>3}$, 
$(9,6,6,5,2)_{4>3}$, 
$(9,7,4,4,4)_{2>1}$, 
$(9,7,5,4,3)_{7>2}$, 
$(9,7,5,5,2)_{5>1}$, 
$(9,7,6,3,3)_{5>3}$, 
$(9,7,6,4,2)_{10>5}$, 
$(9,7,6,5,1)_{6>4}$, 
$(9,7,7,3,2)_{5>1}$, 
$(9,7,7,4,1)_{5>2}$, 
$(9,7,7,5,0)_{2>1}$, 
$(9,8,4,4,3)_{5>2}$, 
$(9,8,5,3,3)_{4>1}$, 
$(9,8,5,4,2)_{11>5}$, 
$(9,8,5,5,1)_{4>3}$, 
$(9,8,6,3,2)_{11>6}$, 
$(9,8,6,4,1)_{12>11}$, 
$(9,8,7,2,2)_{5>3}$, 
$(9,8,7,3,1)_{8>6}$, 
$(9,9,4,3,3)_{3>1}$, 
$(9,9,4,4,2)_{2>1}$, 
$(9,9,5,3,2)_{7>5}$, 
$(9,9,5,4,1)_{6>4}$, 
$(10,5,5,5,3)_{1>0}$, 
$(10,6,4,4,4)_{7>2}$, 
$(10,6,5,4,3)_{6>2}$, 
$(10,6,5,5,2)_{2>0}$, 
$(10,6,6,4,2)_{13>8}$, 
$(10,7,4,4,3)_{8>4}$, 
$(10,7,5,3,3)_{7>3}$, 
$(10,7,5,4,2)_{14>6}$, 
$(10,7,5,5,1)_{6>2}$, 
$(10,7,6,3,2)_{14>8}$, 
$(10,7,6,4,1)_{15>13}$, 
$(10,7,7,2,2)_{1>0}$, 
$(10,7,7,3,1)_{10>5}$, 
$(10,8,4,3,3)_{2>1}$, 
$(10,8,4,4,2)_{17>9}$, 
$(10,8,5,3,2)_{15>8}$, 
$(10,8,5,4,1)_{17>14}$, 
$(10,8,6,2,2)_{17>10}$, 
$(10,9,4,3,2)_{10>7}$, 
$(10,9,4,4,1)_{10>9}$, 
$(10,9,5,2,2)_{10>6}$, 
$(10,10,4,2,2)_{9>5}$, 
$(11,5,4,4,4)_{2>1}$, 
$(11,5,5,4,3)_{3>0}$, 
$(11,6,4,4,3)_{8>4}$, 
$(11,6,5,3,3)_{3>2}$, 
$(11,6,5,4,2)_{13>6}$, 
$(11,6,5,5,1)_{3>2}$, 
$(11,6,6,3,2)_{10>9}$, 
$(11,7,4,3,3)_{6>3}$, 
$(11,7,4,4,2)_{14>9}$, 
$(11,7,5,3,2)_{18>9}$, 
$(11,7,5,4,1)_{18>15}$, 
$(11,7,6,2,2)_{12>7}$, 
$(11,7,7,2,1)_{8>7}$, 
$(11,8,4,3,2)_{17>10}$, 
$(11,8,5,2,2)_{17>12}$, 
$(11,9,3,3,2)_{5>3}$, 
$(11,9,4,2,2)_{12>9}$, 
$(11,10,3,2,2)_{6>4}$, 
$(12,4,4,4,4)_{4>3}$, 
$(12,5,4,4,3)_{4>2}$, 
$(12,5,5,3,3)_{3>0}$, 
$(12,5,5,4,2)_{3>1}$, 
$(12,5,5,5,1)_{1>0}$, 
$(12,6,4,4,2)_{17>11}$, 
$(12,6,5,3,2)_{12>8}$, 
$(12,6,5,4,1)_{13>12}$, 
$(12,6,6,2,2)_{13>10}$, 
$(12,7,3,3,3)_{1>0}$, 
$(12,7,4,3,2)_{17>11}$, 
$(12,7,5,2,2)_{14>10}$, 
$(12,8,3,3,2)_{4>3}$, 
$(12,8,4,2,2)_{23>18}$, 
$(12,9,3,2,2)_{9>8}$, 
$(13,5,4,3,3)_{2>0}$, 
$(13,5,4,4,2)_{8>6}$, 
$(13,5,5,3,2)_{4>2}$, 
$(13,5,5,4,1)_{4>3}$, 
$(13,6,4,3,2)_{13>11}$, 
$(13,6,5,2,2)_{13>11}$, 
$(13,7,3,3,2)_{5>3}$, 
$(13,7,4,2,2)_{16>14}$, 
$(13,8,3,2,2)_{12>11}$, 
$(14,5,4,3,2)_{7>5}$, 
$(15,5,3,3,2)_{1>0}$

\subsection*{$d=4$, $\delta=8$:}
$(7,7,7,7,4)_{1>0}$, 
$(8,6,6,6,6)_{2>1}$, 
$(8,7,6,6,5)_{1>0}$, 
$(8,7,7,5,5)_{3>0}$, 
$(8,7,7,6,4)_{1>0}$, 
$(8,7,7,7,3)_{2>0}$, 
$(8,8,6,6,4)_{7>1}$, 
$(8,8,7,5,4)_{3>0}$, 
$(8,8,7,6,3)_{5>0}$, 
$(8,8,8,4,4)_{8>2}$, 
$(8,8,8,5,3)_{2>1}$, 
$(8,8,8,6,2)_{7>2}$, 
$(9,6,6,6,5)_{2>1}$, 
$(9,7,6,5,5)_{3>0}$, 
$(9,7,6,6,4)_{5>1}$, 
$(9,7,7,5,4)_{7>0}$, 
$(9,7,7,6,3)_{6>0}$, 
$(9,7,7,7,2)_{3>0}$, 
$(9,8,5,5,5)_{1>0}$, 
$(9,8,6,5,4)_{14>2}$, 
$(9,8,6,6,3)_{12>3}$, 
$(9,8,7,4,4)_{10>1}$, 
$(9,8,7,5,3)_{18>2}$, 
$(9,8,7,6,2)_{13>2}$, 
$(9,8,7,7,1)_{3>0}$, 
$(9,8,8,4,3)_{11>2}$, 
$(9,8,8,5,2)_{12>4}$, 
$(9,8,8,6,1)_{7>4}$, 
$(9,9,5,5,4)_{6>0}$, 
$(9,9,6,4,4)_{5>0}$, 
$(9,9,6,5,3)_{15>3}$, 
$(9,9,6,6,2)_{5>2}$, 
$(9,9,7,4,3)_{14>1}$, 
$(9,9,7,5,2)_{17>3}$, 
$(9,9,7,6,1)_{7>2}$, 
$(9,9,7,7,0)_{2>0}$, 
$(9,9,8,3,3)_{8>1}$, 
$(9,9,8,4,2)_{8>1}$, 
$(9,9,8,5,1)_{9>3}$, 
$(9,9,9,3,2)_{3>1}$, 
$(9,9,9,4,1)_{3>0}$, 
$(10,6,6,6,4)_{9>3}$, 
$(10,7,5,5,5)_{3>0}$, 
$(10,7,6,5,4)_{15>1}$, 
$(10,7,6,6,3)_{13>3}$, 
$(10,7,7,4,4)_{5>0}$, 
$(10,7,7,5,3)_{19>1}$, 
$(10,7,7,6,2)_{8>0}$, 
$(10,7,7,7,1)_{4>0}$, 
$(10,8,5,5,4)_{7>0}$, 
$(10,8,6,4,4)_{31>4}$, 
$(10,8,6,5,3)_{32>5}$, 
$(10,8,6,6,2)_{29>8}$, 
$(10,8,7,4,3)_{35>5}$, 
$(10,8,7,5,2)_{34>6}$, 
$(10,8,7,6,1)_{18>6}$, 
$(10,8,8,3,3)_{4>1}$, 
$(10,8,8,4,2)_{33>9}$, 
$(10,8,8,5,1)_{15>9}$, 
$(10,9,5,4,4)_{15>1}$, 
$(10,9,5,5,3)_{16>1}$, 
$(10,9,6,4,3)_{39>6}$, 
$(10,9,6,5,2)_{38>8}$, 
$(10,9,6,6,1)_{16>9}$, 
$(10,9,7,3,3)_{21>5}$, 
$(10,9,7,4,2)_{43>8}$, 
$(10,9,7,5,1)_{28>9}$, 
$(10,9,8,3,2)_{24>7}$, 
$(10,9,8,4,1)_{24>10}$, 
$(10,9,9,2,2)_{2>0}$, 
$(10,9,9,3,1)_{8>3}$, 
$(10,10,4,4,4)_{12>2}$, 
$(10,10,5,4,3)_{18>3}$, 
$(10,10,5,5,2)_{7>0}$, 
$(10,10,6,3,3)_{8>2}$, 
$(10,10,6,4,2)_{42>10}$, 
$(10,10,6,5,1)_{18>7}$, 
$(10,10,6,6,0)_{11>10}$, 
$(10,10,7,3,2)_{23>6}$, 
$(10,10,7,4,1)_{26>12}$, 
$(10,10,8,2,2)_{17>5}$, 
$(10,10,8,3,1)_{13>9}$, 
$(10,10,9,2,1)_{6>4}$, 
$(11,6,6,5,4)_{9>1}$, 
$(11,6,6,6,3)_{10>4}$, 
$(11,7,5,5,4)_{11>0}$, 
$(11,7,6,4,4)_{22>3}$, 
$(11,7,6,5,3)_{31>4}$, 
$(11,7,6,6,2)_{19>6}$, 
$(11,7,7,4,3)_{25>3}$, 
$(11,7,7,5,2)_{25>2}$, 
$(11,7,7,6,1)_{11>2}$, 
$(11,7,7,7,0)_{2>0}$, 
$(11,8,5,4,4)_{26>3}$, 
$(11,8,5,5,3)_{23>2}$, 
$(11,8,6,4,3)_{60>11}$, 
$(11,8,6,5,2)_{58>13}$, 
$(11,8,6,6,1)_{24>13}$, 
$(11,8,7,3,3)_{26>4}$, 
$(11,8,7,4,2)_{64>14}$, 
$(11,8,7,5,1)_{40>15}$, 
$(11,8,8,3,2)_{28>9}$, 
$(11,8,8,4,1)_{30>17}$, 
$(11,9,4,4,4)_{11>1}$, 
$(11,9,5,4,3)_{45>6}$, 
$(11,9,5,5,2)_{33>5}$, 
$(11,9,6,3,3)_{36>8}$, 
$(11,9,6,4,2)_{78>19}$, 
$(11,9,6,5,1)_{46>20}$, 
$(11,9,7,3,2)_{57>14}$, 
$(11,9,7,4,1)_{58>24}$, 
$(11,9,8,2,2)_{20>7}$, 
$(11,9,8,3,1)_{37>21}$, 
$(11,9,9,2,1)_{9>5}$, 
$(11,10,4,4,3)_{21>5}$, 
$(11,10,5,3,3)_{20>4}$, 
$(11,10,5,4,2)_{52>12}$, 
$(11,10,5,5,1)_{20>7}$, 
$(11,10,6,3,2)_{56>16}$, 
$(11,10,6,4,1)_{56>29}$, 
$(11,10,7,2,2)_{30>9}$, 
$(11,10,7,3,1)_{46>26}$, 
$(11,10,8,2,1)_{25>20}$, 
$(11,11,4,3,3)_{10>2}$, 
$(11,11,4,4,2)_{10>3}$, 
$(11,11,5,3,2)_{26>7}$, 
$(11,11,5,4,1)_{23>12}$, 
$(11,11,6,2,2)_{13>5}$, 
$(11,11,6,3,1)_{30>18}$, 
$(11,11,7,2,1)_{19>15}$, 
$(12,6,5,5,4)_{4>0}$, 
$(12,6,6,4,4)_{21>3}$, 
$(12,6,6,5,3)_{14>3}$, 
$(12,6,6,6,2)_{17>8}$, 
$(12,7,5,4,4)_{19>1}$, 
$(12,7,5,5,3)_{22>1}$, 
$(12,7,6,4,3)_{49>10}$, 
$(12,7,6,5,2)_{46>9}$, 
$(12,7,6,6,1)_{17>10}$, 
$(12,7,7,3,3)_{23>3}$, 
$(12,7,7,4,2)_{32>5}$, 
$(12,7,7,5,1)_{26>6}$, 
$(12,8,4,4,4)_{25>5}$, 
$(12,8,5,4,3)_{56>8}$, 
$(12,8,5,5,2)_{32>5}$, 
$(12,8,6,3,3)_{32>7}$, 
$(12,8,6,4,2)_{109>29}$, 
$(12,8,6,5,1)_{54>27}$, 
$(12,8,7,3,2)_{62>17}$, 
$(12,8,7,4,1)_{65>31}$, 
$(12,8,8,2,2)_{30>13}$, 
$(12,8,8,3,1)_{27>20}$, 
$(12,9,4,4,3)_{33>6}$, 
$(12,9,5,3,3)_{35>7}$, 
$(12,9,5,4,2)_{80>18}$, 
$(12,9,5,5,1)_{32>11}$, 
$(12,9,6,3,2)_{88>28}$, 
$(12,9,6,4,1)_{88>45}$, 
$(12,9,7,2,2)_{41>14}$, 
$(12,9,7,3,1)_{71>40}$, 
$(12,9,8,2,1)_{34>28}$, 
$(12,10,4,3,3)_{14>4}$, 
$(12,10,4,4,2)_{52>16}$, 
$(12,10,5,3,2)_{63>18}$, 
$(12,10,5,4,1)_{62>32}$, 
$(12,10,6,2,2)_{60>23}$, 
$(12,10,6,3,1)_{71>48}$, 
$(12,10,7,2,1)_{50>41}$, 
$(12,11,3,3,3)_{2>0}$, 
$(12,11,4,3,2)_{32>11}$, 
$(12,11,4,4,1)_{25>16}$, 
$(12,11,5,2,2)_{32>14}$, 
$(12,11,5,3,1)_{46>31}$, 
$(12,11,6,2,1)_{41>38}$, 
$(12,12,3,3,2)_{3>2}$, 
$(12,12,4,2,2)_{19>10}$, 
$(12,12,4,3,1)_{13>11}$, 
$(13,5,5,5,4)_{1>0}$, 
$(13,6,5,4,4)_{15>1}$, 
$(13,6,5,5,3)_{9>0}$, 
$(13,6,6,4,3)_{26>7}$, 
$(13,6,6,5,2)_{24>8}$, 
$(13,7,4,4,4)_{17>4}$, 
$(13,7,5,4,3)_{45>7}$, 
$(13,7,5,5,2)_{28>3}$, 
$(13,7,6,3,3)_{30>8}$, 
$(13,7,6,4,2)_{73>21}$, 
$(13,7,6,5,1)_{39>18}$, 
$(13,7,7,3,2)_{34>7}$, 
$(13,7,7,4,1)_{36>15}$, 
$(13,7,7,5,0)_{12>11}$, 
$(13,8,4,4,3)_{38>9}$, 
$(13,8,5,3,3)_{33>6}$, 
$(13,8,5,4,2)_{88>23}$, 
$(13,8,5,5,1)_{32>13}$, 
$(13,8,6,3,2)_{91>31}$, 
$(13,8,6,4,1)_{91>55}$, 
$(13,8,7,2,2)_{43>17}$, 
$(13,8,7,3,1)_{65>41}$, 
$(13,9,4,3,3)_{25>6}$, 
$(13,9,4,4,2)_{55>18}$, 
$(13,9,5,3,2)_{85>28}$, 
$(13,9,5,4,1)_{78>41}$, 
$(13,9,6,2,2)_{62>26}$, 
$(13,9,6,3,1)_{94>67}$, 
$(13,9,7,2,1)_{59>50}$, 
$(13,10,3,3,3)_{4>1}$, 
$(13,10,4,3,2)_{55>21}$, 
$(13,10,4,4,1)_{46>33}$, 
$(13,10,5,2,2)_{57>24}$, 
$(13,10,5,3,1)_{75>54}$, 
$(13,10,6,2,1)_{69>68}$, 
$(13,11,3,3,2)_{15>6}$, 
$(13,11,4,2,2)_{32>17}$, 
$(13,11,4,3,1)_{44>37}$, 
$(13,12,3,2,2)_{13>8}$, 
$(13,13,2,2,2)_{1>0}$, 
$(14,5,5,4,4)_{2>0}$, 
$(14,5,5,5,3)_{3>0}$, 
$(14,6,4,4,4)_{18>4}$, 
$(14,6,5,4,3)_{26>4}$, 
$(14,6,5,5,2)_{11>1}$, 
$(14,6,6,3,3)_{8>4}$, 
$(14,6,6,4,2)_{45>17}$, 
$(14,6,6,5,1)_{17>13}$, 
$(14,7,4,4,3)_{31>9}$, 
$(14,7,5,3,3)_{29>6}$, 
$(14,7,5,4,2)_{63>17}$, 
$(14,7,5,5,1)_{24>8}$, 
$(14,7,6,3,2)_{63>23}$, 
$(14,7,6,4,1)_{62>40}$, 
$(14,7,7,2,2)_{14>4}$, 
$(14,7,7,3,1)_{38>21}$, 
$(14,8,4,3,3)_{18>4}$, 
$(14,8,4,4,2)_{66>24}$, 
$(14,8,5,3,2)_{78>27}$, 
$(14,8,5,4,1)_{76>47}$, 
$(14,8,6,2,2)_{70>33}$, 
$(14,8,6,3,1)_{83>68}$, 
$(14,9,3,3,3)_{5>2}$, 
$(14,9,4,3,2)_{63>26}$, 
$(14,9,4,4,1)_{52>38}$, 
$(14,9,5,2,2)_{61>29}$, 
$(14,9,5,3,1)_{85>65}$, 
$(14,10,3,3,2)_{15>6}$, 
$(14,10,4,2,2)_{57>30}$, 
$(14,10,4,3,1)_{56>53}$, 
$(14,11,3,2,2)_{22>14}$, 
$(14,12,2,2,2)_{11>9}$, 
$(15,5,4,4,4)_{6>2}$, 
$(15,5,5,4,3)_{8>0}$, 
$(15,5,5,5,2)_{2>0}$, 
$(15,6,4,4,3)_{22>7}$, 
$(15,6,5,3,3)_{12>3}$, 
$(15,6,5,4,2)_{38>12}$, 
$(15,6,5,5,1)_{10>4}$, 
$(15,6,6,3,2)_{31>17}$, 
$(15,6,6,4,1)_{30>26}$, 
$(15,7,4,3,3)_{18>5}$, 
$(15,7,4,4,2)_{45>20}$, 
$(15,7,5,3,2)_{57>21}$, 
$(15,7,5,4,1)_{54>35}$, 
$(15,7,6,2,2)_{40>20}$, 
$(15,7,6,3,1)_{57>49}$, 
$(15,7,7,2,1)_{25>23}$, 
$(15,8,3,3,3)_{2>0}$, 
$(15,8,4,3,2)_{58>26}$, 
$(15,8,4,4,1)_{49>42}$, 
$(15,8,5,2,2)_{59>33}$, 
$(15,8,5,3,1)_{74>64}$, 
$(15,9,3,3,2)_{19>9}$, 
$(15,9,4,2,2)_{51>32}$, 
$(15,10,3,2,2)_{28>19}$, 
$(16,4,4,4,4)_{7>4}$, 
$(16,5,4,4,3)_{10>3}$, 
$(16,5,5,3,3)_{6>0}$, 
$(16,5,5,4,2)_{8>2}$, 
$(16,5,5,5,1)_{2>0}$, 
$(16,6,4,3,3)_{7>3}$, 
$(16,6,4,4,2)_{36>18}$, 
$(16,6,5,3,2)_{30>14}$, 
$(16,6,5,4,1)_{29>22}$, 
$(16,6,6,2,2)_{27>18}$, 
$(16,7,3,3,3)_{3>0}$, 
$(16,7,4,3,2)_{42>21}$, 
$(16,7,5,2,2)_{36>22}$, 
$(16,7,5,3,1)_{54>50}$, 
$(16,8,3,3,2)_{13>7}$, 
$(16,8,4,2,2)_{53>37}$, 
$(16,9,3,2,2)_{26>20}$, 
$(17,4,4,4,3)_{5>4}$, 
$(17,5,4,3,3)_{4>0}$, 
$(17,5,4,4,2)_{15>9}$, 
$(17,5,5,3,2)_{8>3}$, 
$(17,5,5,4,1)_{8>5}$, 
$(17,6,4,3,2)_{26>17}$, 
$(17,6,5,2,2)_{24>19}$, 
$(17,7,3,3,2)_{10>5}$, 
$(17,7,4,2,2)_{33>27}$, 
$(17,8,3,2,2)_{24>22}$, 
$(18,4,4,4,2)_{9>8}$, 
$(18,5,4,3,2)_{11>7}$, 
$(18,6,3,3,2)_{4>3}$, 
$(19,5,3,3,2)_{1>0}$

\subsection*{$d=5$, $\delta=7$, $\la_1 \leq 8$:}
$(8,7,7,5,5,3)_{1>0}$, 
$(8,7,7,6,4,3)_{1>0}$, 
$(8,7,7,6,5,2)_{1>0}$, 
$(8,7,7,7,3,3)_{1>0}$, 
$(8,8,7,5,4,3)_{2>1}$, 
$(8,8,7,6,3,3)_{1>0}$, 
$(8,8,7,6,4,2)_{3>2}$, 
$(8,8,7,6,5,1)_{2>1}$, 
$(8,8,7,7,4,1)_{1>0}$

\end{document}